\documentclass[12pt,onecolumn,draftcls]{IEEEtran}

\usepackage[dvips]{epsfig}
\usepackage[usenames,dvipsnames]{pstricks}
\usepackage{algorithm,algorithmic}
\usepackage{amsbsy}
\usepackage{amscd}
\usepackage{amsmath,amssymb,amsfonts,bm}
\usepackage{amsthm}
\usepackage{cite}
\usepackage{color}
\usepackage{enumerate}
\usepackage{graphicx}
\usepackage{listings}
\usepackage{makecell}
\usepackage{pgf}
\usepackage{pifont}
\usepackage{subfigure}
\usepackage{tikz,pgf}
\usepackage{xpatch,letltxmacro}
\tikzset{font={\fontsize{10pt}{12}\selectfont}}
\newcommand\scalemath[2]{\scalebox{#1}{\mbox{\ensuremath{\displaystyle #2}}}}

\newtheorem{lemma}{Lemma}

\newtheorem{remark}{Remark}

\def\bA{\mathbf{A}}

\def\bB{\mathbf{B}}

\def\bI{\mathbf{I}}

\def\complexC{{\mathbb{C}}}

\def\realR{{\mathbb{R}}}

\begin{document}
\title{Design of a Multi-User Wireless Powered Communication System Employing Either Active IRS or AF Relay}

\author{Omid Rezaei, Maryam Masjedi, Ali Kanaani, Mohammad Mahdi Naghsh\IEEEauthorrefmark{1}, Saeed~Gazor, and Mohammad Mahdi Nayebi
	
	\thanks{O. Rezaei and M. M. Nayebi are with the Department of Electrical Engineering, Sharif University of Technology, Tehran, 11155-4363, Iran. M. Masjedi, A. Kanaani, and M. M. Naghsh are with the Department of Electrical and Computer Engineering, Isfahan University of Technology, Isfahan, 84156-83111, Iran. S. Gazor is with the Department of Electrical and Computer Engineering, Queen’s University, Kingston, Ontario, K7L 3N6, Canada.
		
		\IEEEauthorrefmark{1}Please address all the correspondence to M. M. Naghsh, Phone: (+98) 31-33912450; Fax: (+98) 31-33912451; Email: mm\_naghsh@cc.iut.ac.ir}
}
\maketitle

\begin{abstract}
 In this paper, we optimize a Wireless Powered Communication (WPC) system including multiple pair of users, where transmitters employ single-antenna to transmit their information and power to their receivers with the help of one multiple-antennas Amplify-and-Forward (AF) relay or an active Intelligent Reflecting Surface (IRS). We propose a joint Time Switching (TS) scheme in which transmitters, receivers, and the relay/IRS are either in their energy or information transmission/reception modes. The transmitted multi-carrier unmodulated and modulated waveforms are used for Energy Harvesting (EH) and Information Decoding (ID) modes, respectively. In order to design an optimal fair system, we maximize the minimum rate of all pairs for both relay and IRS systems through a unified framework.
This framework allows us to simultaneously design energy waveforms, find optimal relay/IRS amplification/reflection matrices, allocate powers for information waveforms, and allocate time durations for various phases. In addition, we take into account the non-linearity of the EH circuits in our problem. This problem turns out to be non-convex. Thus, we propose an iterative algorithm by using the Minorization-Maximization (MM) technique, which quickly converges to the optimal solution. Numerical examples show that the proposed method improves the performance of the multi-pair WPC relay/IRS system under various setups.
\end{abstract}

\begin{IEEEkeywords}
Fair Throughput Maximization, Intelligent Reflecting Surface (IRS), Minorization-Maximization (MM), Relay Networks, Wireless Powered Communication (WPC).
\end{IEEEkeywords}


\IEEEpeerreviewmaketitle

\section{Introduction}
\IEEEPARstart{W}{ireless} Power Transfer (WPT) technology is introduced to extend the lifetime of devices in wireless networks in which the energy is emitted from the dedicated power sources to the devices~\cite{ho2012optimal}.
Interestingly, WPT enables Simultaneous Wireless Information and Power Transfer (SWIPT)~\cite{perera2017simultaneous}, where devices not only receive and decode information, but also harvest the energy from Radio Frequency (RF) signals. The Time Switching (TS) and Power Splitting (PS) schemes are two well-known implementing protocols of SWIPT~\cite{liu2012wireless}.
Recently, SWIPT models are designed to employ relays to further enhance the coverage and spectral efficiency of wireless networks~\cite{rostampoor2017energy,lin2018source,chen2015joint}. In~\cite{rostampoor2017energy}, a Multiple-Input Multiple-Output (MIMO) two-way relay system is introduced in which two transceivers exchange their information through a relay. The authors in~\cite{lin2018source} designed the relay and source precoders by minimizing the bit error rate at the destination for a full-duplex MIMO relay system and SWIPT-enabled users. A similar system with a half-duplex two-way relay is designed in~\cite{chen2015joint} by minimizing the mean square error at the destination.

Another impactful technology that is currently emerging is to use of the Intelligent Reflecting Surface (IRS) in wireless communication systems. This promising solution not only is capable of improving energy delivery but also can enhance the spectral efficiency of future wireless communication networks.~\cite{renzo2019smart}. An IRS is an array of large number of Reflecting Elements (REs) designed to have controllable electromagnetic properties. Each RE introduces a phase shift on the impinging signal, allowing beamforming/manipulation of the reflection waveforms. Precisely, the IRS matrix allows controlling the reflected signal (amplify, attenuated, steer in the desired direction, and so on) toward optimal desired directions by purposefully designing the phase shift matrix. The IRS is exploited recently in SWIPT systems in~\cite{wu2019weighted,khalili2021multi,wu2020joint,pan2020intelligent}. The weighted sum harvested power maximization problem was studied in~\cite{wu2019weighted} for an IRS-aided SWIPT model in which a multiple-antennas access point serves multiple single-antenna users. In~\cite{khalili2021multi}, the model in~\cite{wu2019weighted} is extended for a more practical multi-objective optimization problem by taking into account the trade-off between sum rate and sum harvested power maximization. In~\cite{wu2020joint}, the total transmission power is minimized for a Multiple-Input Single-Output (MISO) SWIPT system employing multiple IRSs. In~\cite{pan2020intelligent}, the MISO SWIPT in~\cite{wu2019weighted,khalili2021multi,wu2020joint} is extended to the MIMO case, and the weighted sum rate is maximized in an IRS-assisted system.

The design of the energy waveform remarkably affects the  performance of WPT-based systems. Indeed, an efficient waveform leads to significant improvement in the efficacy of power delivery. Experiments reveal that signals with high Peak to Average Power Ratio (PAPR) such as multi-sine signals provide more DC power at Energy Harvesting (EH) circuits than constant envelope signals with the same average RF power~\cite{boaventura2013optimum}. Based on this interesting observation, a multi-sine waveform design for WPT has been examined in several works~\cite{clerckx57,clerckx60,clerckx48,clerckx59,zhao2021irs,clerckx58,clerckx69,clerckx70,clerckx2017wireless}.
Waveform design with a non-linear EH model was considered in~\cite{clerckx57} and~\cite{clerckx60} for MISO and Single-Input Single-Output (SISO) systems, respectively. The authors of~\cite{clerckx48} proposed a low-complexity method for a waveform design in a SISO WPT system.
In~\cite{clerckx59}, the transmit waveform was designed based on limited Channel State Information (CSI) feedback WPT system. Then, the authors in~\cite{zhao2021irs} studied waveform design for an IRS-aided SWIPT MISO system.
The aforementioned methods for design of single-user systems were extended to the multi-user case in \cite{clerckx58}.
Also, a waveform design was performed in~\cite{clerckx69} for wireless powered backscatter communication networks, and this work was extended to multi-user backscatter systems in~\cite{clerckx70}. The authors of~\cite{clerckx2017wireless} investigated the waveform and transceiver design problem in a MISO SWIPT system and determined the multi-sine waveforms for Information Decoding (ID) and EH phases.

In this paper, we  optimize a multi-user wireless powered relay/IRS system using a multi-sine waveform with the following main contributions:
\begin{itemize}
 \item \textit{Relay model}:  To the best of our knowledge, a multi-sine waveform design for multiple user pairs in a wireless powered relay system has not been addressed in the literature. In this paper, we consider a multi-carrier Wireless Powered Communication (WPC) system for multi-user relay channels. Precisely, in our proposed model, an amplify-and-forward (AF) relay provides energy/information transmission from ${K}$ transmitters to their receivers by adopting the TS scheme in all nodes\footnote{Note that the system in~\cite{lee2018joint}, where a TS scheme is only applied for a receiver, is considered as a special case of the proposed joint TS scheme.}. Herein, the aim is to design the multi-carrier unmodulated energy waveforms and the allocated power for information waveforms at the transmitters, the amplification matrices in a relay, and the time durations for the EH and ID modes in order to maximize the minimum rate of the user pairs. In addition, we consider the effect of the non-linearity of EH circuits in the design problem.

 \item \textit{IRS model}: In the case of IRS-assisted communication, multi-pair WPC has not been considered in the literature, and therefore, herein, we consider this type of IRS-assisted systems. Precisely, in this case, an active IRS\footnote{Note that in an active IRS, REs can amplify the reflected signals using their reflection-type amplifiers~\cite{zhang2021active}.} replaced with the AF relay in the proposed system model mentioned in the above paragraph. Also, some comparisons are made between relay and IRS models in terms of architecture and performance (see Remark~\ref{kl}-\ref{kl2}).

 \item \textit{Unified consideration of relay and IRS}: Both proposed AF relay and active IRS-aided systems are modeled under a unified formulation, and we handle the resulting optimization problems under a unified mathematical umbrella. We show that the problem is non-convex and consequently, is hard to solve. To deal with this problem, we devise a method based on the Minorization-Maximization (MM) technique. 
 	Interestingly, the proposed algorithm can deal with relay and IRS systems by switching between Kronecker and Hadamard products for parameters used in the algorithm (see Lemma~2).
 \item \textit{Sub-optimal methods}: Some sub-optimal methods with lower signaling overhead and computational complexity are proposed and then, their performance are compared.

 \item \textit{Numerical result}: Simulation results are reported to illustrate the effectiveness of the proposed method; particularly, the impact of the relay/IRS matrix and energy waveform design. Also, numerical examples show that the minimum rate of users increases linearly/super-linearly with the number of antennas/REs in relay/IRS systems.
\end{itemize}
The rest of this paper is organized as follows: The signal and system models are explained in Section \ref{sys}. In Section \ref{sum}, the minimum rate maximization problem is formulated, and a unified optimization framework is proposed for both relay and IRS models. Section \ref{num} presents numerical examples to illustrate the effectiveness of the proposed method. Finally, conclusions are drawn in Section \ref{con}.

\emph{Notation:} Bold lowercase (uppercase) letters are used for vectors (matrices).
The notations $\textrm{arg} (\cdot)$, $\mathbb{E} [\cdot ]$, $\Re \{ \cdot \}$, $\|\cdot\|_2$, ${(\cdot)^{{T}}}$, $(\cdot)^{{H}}$, $(\cdot)^{{*}}$, $\mbox{tr} \{ \cdot \}$, $\mathbf{\lambda}_{\textrm{max}}(\cdot)$, $\textrm{vec}(\cdot)$, $\textrm{Diag}(\cdot)$, ${\nabla}_{{\mathbf{x}}} f(\cdot)$ and ${\nabla}^{2}_{{\mathbf{x}}} f(\cdot)$ indicate the phase of a complex number, statistical expectation, real-part, $l_2$-norm of a vector,  transpose, Hermitian,  complex conjugate, trace of a matrix, the principal eigenvalue of a matrix, stacking of the column of a matrix, a diagonal matrix formed by the entries, the gradient of a function with respect to (w.r.t.)  $\mathbf{x}$ and the Hessian of a function w.r.t. $\mathbf{x}$, respectively. The symbols $\otimes$ and $\odot$ stand for the Kronecker and Hadamard products of two matrices. We denote $\mathcal{CN}(\boldsymbol{\omega},\mathbf{\Sigma})$ as a circularly symmetric complex Gaussian (CSCG) distribution with mean $\boldsymbol{\omega}$ and covariance $\mathbf{\Sigma}$. The set $\mathbb{R}_+$ represents non-negative real numbers and $\complexC^{N \times N}$ and $\mathbb{D}^{N \times N}$ are the set of ${N \times N}$ complex and complex diagonal matrices, respectively. The set of ${N \times N}$ positive (semi-)definite and identity matrices are denoted by $\mathbb{S}_{++}^{N}\subset \complexC^{N \times N}$ ($\mathbb{S}_{+}^{N} \subset \complexC^{N \times N}$) and $\bI_{N}$, respectively. The notation $\bA \succ \bB$ ($\bA \succeq\bB$) means that $\bA-\bB$ is positive (semi-)definite.
\section{System Model} \label{sys}
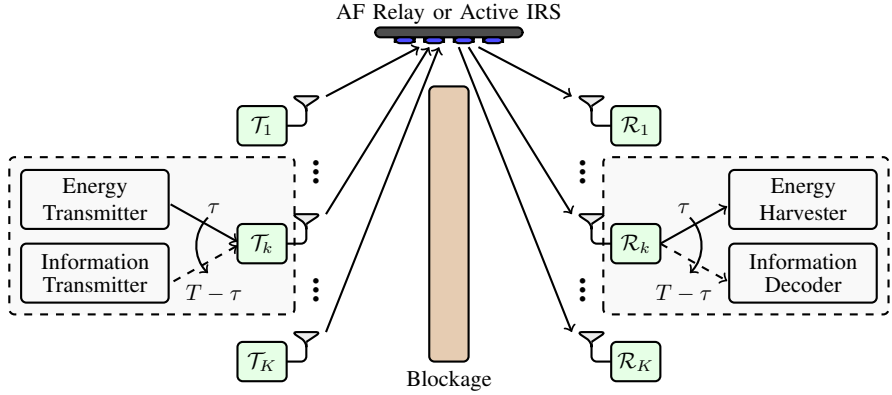
\begin{figure}
 \centering
 \begin{tikzpicture}[even odd rule,rounded corners=2pt,x=12pt,y=12pt,scale=.62,every node/.style={scale=.85}]

 \draw[thick,fill=green!10] (-8.5-1.25,5-1) rectangle ++(2.5,2) node[midway]{$\mathcal{T}_1$};
 \draw[thick] (-6-1.25,6-1)--++(1,0)--+(0,1);
 \draw[thick,fill=gray!15] (-4.5-1.25-.5,7-1)--++(.75,0.5)--++(-1.5,0)--++(.75,-.5)--++(0,-.1);
 \draw[->,thick] (-5.25,6.5)--++(5.25-.5,2.5);
 \draw (-4.5-1.25,3.5-.5) node {\huge$\vdots$};

 \draw[dashed,thick,fill=gray!5] (-5.5-1.35,-4.6) rectangle ++(-14.5,8) ;

 \draw[<-,thick] (-8.5-1.25,-.9)--++(-3.4,1.9);
 \draw[<-,thick,black!100] (-10-1.25,-2.5) arc (-135:-225:2) ;
 \draw[<-,dashed,thick] (-8.5-1.25,-1.1)--++(-3.4,-1.9);

 \node [] at (-9.7-1.25,.8) {$\tau$};
 \node [] at (-9.7-1.25,-3.5) {$T- \tau$};
 \draw[thick] (-12-1.25,-4) rectangle ++(-7.5,3) node[midway]{\shortstack[c]{Information\\ Transmitter }};
 \draw[thick] (-12-1.25,-.25) rectangle ++(-7.5,3) node[midway]{ \shortstack[c]{ Energy\\ Transmitter }};

 \draw[thick,fill=green!10] (-8.5-1.25,-2) rectangle ++(2.5,2) node[midway]{ $\mathcal{T}_k$};
 \draw[thick] (-6-1.25,-1)--++(1,0)--+(0,1);
 \draw[thick,fill=gray!15] (-4.5-1.25-.5,0)--++(.75,0.5)--++(-1.5,0)--++(.75,-.5)--++(0,-.1);
 \draw[->,thick] (-5.25,0.5)--++(5.25,8.5);

 \draw (-4.5-1.25,-3) node {\huge$\vdots$};

 \draw[thick,fill=green!10] (-8.5-1.25,-8) rectangle ++(2.5,2) node[midway]{$\mathcal{T}_{{K}}$};
 \draw[thick] (-6-1.25,-8+1)--++(1,0)--+(0,1);
 \draw[thick,fill=gray!15] (-4.5-1.25-.5,-7+1)--++(.75,0.5)--++(-1.5,0)--++(.75,-.5)--++(0,-.1);
    \draw[->,thick] (-5.25,-5.5)--++(5.25+0.5,14.5);

\draw[thick,fill=black!70] (-2.75,9.5) rectangle ++(7.5,.5);
\draw[thick,fill=blue!70] (-2.75+1,9.2) rectangle ++(1,.3);
\draw[thick,fill=blue!70] (-2.75+1+1.5,9.2) rectangle ++(1,.3);
\draw[thick,fill=blue!70] (-2.75+1+3,9.2) rectangle ++(1,.3);
\draw[thick,fill=blue!70] (-2.75+1+4.5,9.2) rectangle ++(1,.3);
\node[] at (1,10.7) {AF Relay or Active IRS};

 \draw[thick,fill=green!10] (8.5-1.25+2,5-1) rectangle ++(2.5,2) node[midway]{ $\mathcal{R}_1$};
 \draw[thick] (8.5-1.25+2,6-1)--++(-1,0)--+(0,1);
 \draw[thick,fill=gray!15] (7-1.25+.5+2,7-1)--++(.75,0.5)--++(-1.5,0)--++(.75,-.5)--++(0,-.1);
    \draw[<-,thick] (7.25,6.5)--++(-4.75,2.5);
 \draw (7-1.2+2,3) node {\huge$\vdots$};

 \draw[dashed,thick,fill=gray!5] (8-1.15+2,-4.6) rectangle ++(14.5,8) ;

 \draw[thick,fill=green!10] (8.5-1.25+2,-2) rectangle ++(2.5,2) node[midway]{ $\mathcal{R}_k$};
 \draw[thick] (8.5-1.25+2,-1)--++(-1,0)--+(0,1);
 \draw[thick,fill=gray!15] (7-.75+2,0)--++(.75,0.5)--++(-1.5,0)--++(.75,-.5)--++(0,-.1);
 \draw[<-,thick] (7.25,0.5)--++(-5.25,8.5);
 \draw (7-1.2+2,-3) node {\huge$\vdots$};

 \draw[->,thick] (11-1.25+2,-1)--++(3.4,1.9);
 \draw[<-,thick,black!100] (12.5-1.25+2,-2.5) arc (-45:45:2) ;
 \draw[->,dashed,thick] (11-1.25+2,-1)--++(3.4,-1.9);

 \node [] at (12.2-1.25+2,.8) {$\tau$};
 \node [] at (12.2-1.25+2,-3.5) {$T- \tau$};

 \draw[thick] (14.5-1.25+2,-4) rectangle ++(7.5,3) node[midway]{\shortstack[c]{Information \\ Decoder }};
 \draw[thick] (14.5-1.25+2,-.25) rectangle ++(7.5,3) node[midway]{\shortstack[c]{Energy \\ Harvester }};

 \draw[thick,fill=green!10] (8.5-1.25+2,-8) rectangle ++(2.5,2) node[midway]{ $\mathcal{R}_{{K}}$};
 \draw[thick] (8.5-1.25+2,-7)--++(-1,0)--+(0,1);
 \draw[thick,fill=gray!15] (7-.75+2,-6)--++(.75,0.5)--++(-1.5,0)--++(.75,-.5)--++(0,-.1);
 \draw[<-,thick] (7.25,-5.5)--++(-5.25-0.5,14.5);

 \draw[thick,fill=brown!40] (-1.25-1.25+.5+2,-7) rectangle ++(2,14);

 \node[] at (1,-8) { Blockage};
 \end{tikzpicture}
 \caption{Multi-user wireless powered relay/IRS system based on TS scheme with blocked direct path.}
 \label{ht}
 \centering
\end{figure}
We consider a multi-carrier wireless powered relay/IRS system with $K$ user pairs $\{(\mathcal{T}_k,\mathcal{R}_k)\}_{k=1}^K$ as shown in Fig.~\ref{ht}, where the direct links between the transmitters and receivers are likely blocked (see~\cite{shin2016mimo} and~\cite{jiang2022interference,bafghi2020degrees} for similar models of multiple user pairs with blocked direct path for relay and IRS systems, respectively).  The single-antenna transmitter $\mathcal{T}_k$ communicates with its receiver $\mathcal{R}_k$ through either an AF relay with $M_{\mathrm{R}}$ antennas or an active IRS with $M_{\mathrm{IRS}}$ REs. We assume that $\mathcal{R}_k$ harvests a part of its required power, whereas $\mathcal{T}_k$ and the relay/IRS have no energy concern~\cite{rostampoor2017energy,lin2018source}.
In each time duration $T$, the relay/IRS helps  $\mathcal{R}_k$ not only harvest energy from the signal of all $\{\mathcal{T}_k\}_{k=1}^K$, but also decode the information from its corresponding transmitter $\mathcal{T}_k$ by using a joint TS scheme. Precisely, $\mathcal{T}_k$, relay/IRS, and $\mathcal{R}_k$ switch simultaneously at time $t=\tau$ from their energy delivery modes to their communication modes. We assume that all nodes are perfectly synchronized as shown in Fig.~\ref{kkj} for this switching~\cite{lee2018joint}. We consider a multi-carrier system with a total bandwidth of $B_t$ equally divided into ${N}$ orthogonal subbands. We also model all channels to have a frequency-selective block fading, i.e., the channel coefficients remain constant for at least $T$ seconds.
Let the complex random matrices $\mathbf{H}^{{\mathrm{R}}}_{n} \in \complexC^{{M_{\mathrm{R}}} \times {K}}$ and $\mathbf{G}^{{\mathrm{R}}}_{n} \in \complexC^{ {K} \times {M_{\mathrm{R}}}}$ denote the channels from transmitters to the relay and the channels from the relay to the receivers for $n^{th}$ subband, respectively.
The elements of $\mathbf{H}^{{\mathrm{R}}}_n$ and $\mathbf{G}^{{\mathrm{R}}}_n$ are zero mean CSCG random variables in the case of Rayleigh fading. Similarly, we denote the channels from transmitters to the IRS and the channels from the IRS to the receivers by $\mathbf{H}^{{\mathrm{IRS}}}_{n} \in \complexC^{{M_{\mathrm{IRS}}} \times {K}} $ and $\mathbf{G}^{{\mathrm{IRS}}}_{n} \in \complexC^{ {K} \times {M_{\mathrm{IRS}}}} $, respectively. In the sequel, $\mathbf{H}_{n}$ and $\mathbf{G}_{n}$ refer to either $\mathbf{H}^{{\mathrm{R}}}_{n}$ and $\mathbf{G}^{{\mathrm{R}}}_{n}$ or $\mathbf{H}^{{\mathrm{IRS}}}_{n}$ and $\mathbf{G}^{{\mathrm{IRS}}}_{n}$, depending on the case under discussion. In addition, we assume that we can control the relay and IRS by collecting and using the CSI of all links \cite{jiang2022interference,bafghi2020degrees,cheng2015degrees}.
For example, a relay itself can act as a controller. The CSI may be estimated in various ways, e.g., by using orthogonal pilot sequences ( see~\cite{wang2020channel,kudathanthirige2017massive} for more details). The CSI estimation is out of the scope of this paper.
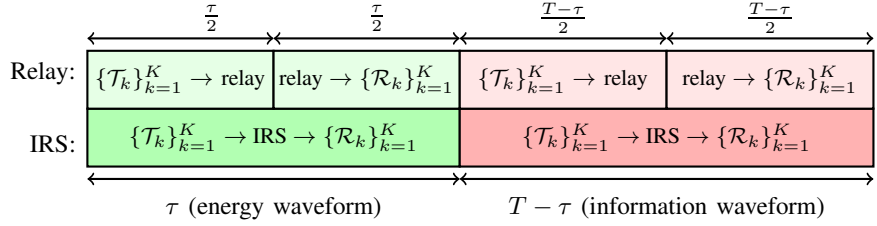
\begin{figure}
 \centering{
 \begin{tikzpicture}[scale=1.1,every node/.style={scale=1.05}]
 \draw[<->,thick] (-.5,.85)--++(2.25,0);
 \node [left] at (-.5,.425) {\footnotesize Relay:};
\node [left] at (-.5,-.425) {\footnotesize IRS:};
 \draw[<->,thick] (-.5,-.85)--++(4.5,0);
 \draw[<->,thick] (1.75,.85)--++(2.25,0);
 \draw[<->,thick] (4,.85)--++(2.5,0);
 \draw[<->,thick] (4,-.85)--++(5,0);
 \draw[<->,thick] (6.5,.85)--++(2.5,0);
 \draw[thick,fill=green!30] (-.5,-.7) rectangle ++(4.5,.7) node[midway]{\scriptsize $\{\mathcal{T}_k\}_{k=1}^K\rightarrow\hspace{2pt}$IRS$\hspace{2pt}\rightarrow \{\mathcal{R}_k\}_{k=1}^K$} ;
 \draw[thick,fill=green!10] (-.5,0) rectangle ++(2.25,.7) node[midway]{\scriptsize $\{\mathcal{T}_k\}_{k=1}^K \rightarrow$\hspace{3pt}relay};
 \draw[thick,fill=green!10] (1.75,0) rectangle ++(2.25,.7) node[midway]{\scriptsize relay\hspace{2pt}$\rightarrow\{\mathcal{R}_k\}_{k=1}^K$} ;
 \draw[thick,fill=red!30] (4,-.7) rectangle ++(5,.7) node[midway]{\scriptsize $\{\mathcal{T}_k\}_{k=1}^K \rightarrow\hspace{2pt}$IRS$\hspace{2pt}\rightarrow \{\mathcal{R}_k\}_{k=1}^K$} ;
 \draw[thick,fill=red!10] (4,0) rectangle ++(2.5,.7) node[midway]{\scriptsize $\{\mathcal{T}_k\}_{k=1}^K \rightarrow$\hspace{3pt}relay};
 \draw[thick,fill=red!10] (6.5,0) rectangle ++(2.5,.7) node[midway]{\scriptsize relay$\hspace{2pt}\rightarrow \{\mathcal{R}_k\}_{k=1}^K$};
 \node [] at (1,1.1) {\footnotesize $\frac{\tau}{2}$};
 \node [] at (1.75,-1.2) {\footnotesize $\tau$ (energy waveform)};
 \node [] at (3,1.1) {\footnotesize $\frac{\tau}{2}$};
 \node [] at (5.25,1.1) {\footnotesize $\frac{T-\tau}{2}$};
 \node [] at (6.5,-1.2) {\footnotesize $T- \tau$ (information waveform)};
 \node [] at (7.75,1.1) {\footnotesize $\frac{T-\tau}{2}$};
 \end{tikzpicture}}
 \caption{The transmission, amplification/reflection, and reception timeline for the proposed relay/IRS model.}
 \label{kkj}
 \centering
\end{figure}
Also, we propose two low-complexity implementation methods mentioned in Remark~\ref{hh}  to reduce the signaling overhead in the controller node.

Each $\mathcal{T}_k$ transmits a multi-sine energy waveform $x_{\mathrm{E},k}(t)$ and a multi-carrier modulated information waveform $x_{\mathrm{I},k}(t)$ to the relay/IRS during the first-hop transmission at the EH and ID time slots, respectively, as follows
\begin{align} \label{eq1}
x_{\mathrm{E},k}(t) &= \sum_{n=1}^{{N}} a_{\mathrm{E},k,n} \textrm{cos} (2 \pi f_n t + \phi_{\mathrm{E},k,n}),
=\Re \left \{ \sum_{n=1}^{{N}} s_{\mathrm{E},k,n} e^{j2 \pi f_n t} \right \} ,
\\ 
\label{eq1mod}
x_{\mathrm{I},k}(t) &= \sum_{n=1}^{{N}} a_{\mathrm{I},k,n}(\tau) \textrm{cos} (2 \pi f_n t + \phi_{\mathrm{I},k,n})=\Re \left \{ \sum_{n=1}^{{N}} s_{\mathrm{I},k,n} e^{j2 \pi f_n t} \right \},
\end{align}
where $s_{\mathrm{E},k,n}= a_{\mathrm{E},k,n} e^{j\phi_{\mathrm{E},k,n}}$ and $s_{\mathrm{I},k,n}= a_{\mathrm{I},k,n} e^{j\phi_{\mathrm{I},k,n}}$ are the baseband complex signal representations for the energy and information waveforms, respectively. We assume that the baseband information signals are i.i.d. CSCG random variable variables, i.e., $s_{\mathrm{I},k,n} \sim \mathcal{CN} (0, p_{\mathrm{I},k,n})$.
 The transmitted energy by $\mathcal{T}_k$  is constrained by
\begin{equation} \label{41}
 \frac{\tau}{2\rho}  \vert  s_{\mathrm{E},k,n} \vert^2  + \frac{T- \tau}{2\rho}  p_{\mathrm{I},k,n}  \leq T p^{\mathrm{rf}}_{{k},n} ,  \hspace{5pt} \forall k,n,
\end{equation}
where $p^{\mathrm{rf}}_{{k},n}$ is the maximum power budget at $\mathcal{T}_k$ for the $n^{th}$ subband and $\rho$ addresses both
$\rho_{\mathrm{R}}=2$ for relay and $ \rho_{\mathrm{IRS}}=1$ for IRS system according to the proposed timeline in Fig.~\ref{kkj} (see also Remark~1). By defining $\mathbf{s}_{\mathrm{E},n} =[s_{\mathrm{E},1,n} ,  \cdots , s_{\mathrm{E},{K},n}]^{T}$ and $\mathbf{s}_{\mathrm{I},n} =[s_{\mathrm{I},1,n} , \cdots , s_{\mathrm{I},{K},n}]^{T}$,
the received signal at the relay/IRS is expressed as
\begin{equation} \label{fg}
\mathbf{r}(t) =
\begin{cases}
\sum\limits_{n=1}^{{N}}\lbrace\mathbf{H}_n \mathbf{s}_{\mathrm{E},n} + \mathbf{z}_{n}\rbrace ,\ t \in T_{\mathrm{EH}}, \ \textrm{for EH}, \\
\sum\limits_{n=1}^{{N}} \lbrace\mathbf{H}_n \mathbf{s}_{\mathrm{I},n}+ \mathbf{z}_{n}\rbrace ,\  t \in T_{\mathrm{ID}}, \ \textrm{for ID},
\end{cases}
\end{equation}
where
\begin{align}
 T_{\mathrm{EH}}&=
 \begin{cases}
 0 \leq t \leq \frac{\tau}{2}, \ \textrm{for relay}, \\
 0 \leq t \leq {\tau}, \   \textrm{for IRS},
 \end{cases}
\hspace{25pt}
 T_{\mathrm{ID}}=
 \begin{cases}
 \tau \leq t \leq \tau + \frac{T- \tau}{2}, \  \textrm{for relay}, \\
 \tau \leq t \leq T, \  \textrm{for IRS},
 \end{cases}
\end{align}
and $\mathbf{r}$ denotes either $\mathbf{r}_{{\mathrm{R}}}$ or $\mathbf{r}_{{\mathrm{IRS}}}$. Furthermore, the AWGN $\mathbf{z}_{n}$ denotes either $\mathbf{z}^{{\mathrm{R}}}_{n}\sim \mathcal{CN} (0, \sigma^{2}_{{{\mathrm{R}}},n} \mathbf{I}_{M_{{\mathrm{R}}}})$ or $\mathbf{z}^{{\mathrm{IRS}}}_{n}\sim \mathcal{CN} (0, \sigma^{2}_{{{\mathrm{IRS},n}}} \mathbf{I}_{M_{{\mathrm{IRS}}}})$ for relaying or reflecting modes. In contrast to the passive IRS, an active IRS adds non-negligible noise (which is introduced by the active elements~\cite{zhang2021active,long2021active}); however, the added noise of an active IRS has considerably less impact compared to the relay noise (which is introduced by RF chains), i.e., $\sigma^{2}_{{{\mathrm{IRS},n}}} \leq \sigma^{2}_{{{\mathrm{R}}},n}$~\cite{bousquet20114}.

In the second-hop transmission, the relay/IRS amplifies the energy and information signals of $\mathcal{T}_k$ by amplification/reflection matrices and then forwards them to $\mathcal{R}_{k}$. For AF relay system, the amplification matrices is introduced as $\mathbf{U}^{\mathrm{R}}_{\mathrm{E},n}$ and $\mathbf{U}^{\mathrm{R}}_{\mathrm{I},n} \in {\complexC}^{{M_{\mathrm{R}}} \times M_{\mathrm{R}}}, \forall n$ for energy and information phases, respectively. In the case of IRS-aided system, the reflection matrices is defined as $\mathbf{U}^{\mathrm{IRS}}_{\mathrm{E}}=\textrm{Diag} (\hm{\theta}_\mathrm{E})$ and $\mathbf{U}^{\mathrm{IRS}}_{\mathrm{I}}=\textrm{Diag} (\hm{\theta}_\mathrm{I})$ for energy and information time slots, respectively, where $\hm{\theta}_\mathrm{E}=[\eta_{\mathrm{E},1} e^{j\theta_{\mathrm{E},1}}, \eta_{\mathrm{E},2} e^{j\theta_{\mathrm{E},2}},\cdots, \eta_{\mathrm{E},M_{\mathrm{IRS}}} e^{j\theta_{\mathrm{E},M_{\mathrm{IRS}}}}]^T$ and $\hm{\theta}_\mathrm{I}=[\eta_{\mathrm{I},1} e^{j\theta_{\mathrm{I},1}}, \eta_{\mathrm{I},2} e^{j\theta_{\mathrm{I},2}},\cdots, \eta_{\mathrm{I},M_{\mathrm{IRS}}} e^{j\theta_{\mathrm{I},M_{\mathrm{IRS}}}}]^T$ with $\eta_{\mathrm{E},m},\eta_{\mathrm{I},m} \geq 1$ and $\theta_{\mathrm{E},m},\theta_{\mathrm{I},m} \in [0, 2\pi]$ respectively denote the reflection amplitude and the phase shift at the $m^{th}$ RE\footnote{Note that passive and passive lossless IRS require $\eta_{\mathrm{E},m},\eta_{\mathrm{I},m} \in [0, 1]$ and $\eta_{\mathrm{E},m}=\eta_{\mathrm{I},m} =1$, respectively.}. 
\begin{remark} \label{kl}
An active IRS amplifies the signal without any significant delay. However, in an AF relay, the signal reception, amplification, and transmission at the RF chain cause a long delay. Therefore, in practice, the AF relay requires twice time compared to the active IRS for transmission one information symbol~\cite{zhang2021active}.
\end{remark}
We define $\mathbf{U}_{\mathrm{E},n}$ and $\mathbf{U}_{\mathrm{I},n}$ to address both $\mathbf{U}^{\mathrm{R}}_{\mathrm{E},n}$, $\mathbf{U}^{\mathrm{IRS}}_{\mathrm{E}}$ and $\mathbf{U}^{\mathrm{R}}_{\mathrm{I},n}$, $\mathbf{U}^{\mathrm{IRS}}_{\mathrm{I}}$, respectively. The forwarded signal by the relay/IRS is given by
\begin{equation*}
 \widetilde{\mathbf{r}}(t) \hspace{-3pt}=
 \begin{cases}
\sum_{n=1}^{{N}} \mathbf{U}_{\mathrm{E},n} \left ( \mathbf{H}_n \mathbf{s}_{\mathrm{E},n} + \mathbf{z}_{n} \right) , ~   t \in \widetilde{T}_{\mathrm{EH}}, ~  \textrm{for~EH} , \\
\sum_{n=1}^{{N}} \mathbf{U}_{\mathrm{I},n} \left ( \mathbf{H}_n \mathbf{s}_{\mathrm{I},n}+ \mathbf{z}_{n}  \right) , ~   t \in \widetilde{T}_{\mathrm{ID}}, ~ \textrm{for ~ID},
\end{cases}
\label{fg54}
\end{equation*}
where
\begin{align}
 \widetilde{T}_{\mathrm{EH}}&=
 \begin{cases}
 \frac{\tau}{2} \leq t \leq \tau, ~ \mbox{for relay}, \\
 0 \leq t \leq {\tau}, ~ \textrm{for~IRS},
 \end{cases}
 \hspace{25pt}
 \widetilde{T}_{\mathrm{ID}}=
 \begin{cases}
 \tau + \frac{T-\tau}{2} \leq t \leq T, ~ \mbox{for relay}, \\
 \tau \leq t \leq T, ~ \textrm{for~IRS},
 \end{cases}
\end{align}
and $\widetilde{\mathbf{r}}$ denotes either $\widetilde{\mathbf{r}}_{\mathrm{R}}$ or $\widetilde{\mathbf{r}}_{\mathrm{IRS}}$ for relay or IRS system, with a slight abuse of notation.

Then, the power of $\mathbf{\widetilde{r}}(t)$ from the relay/IRS is written as
\begin{equation} \label{newq1}
{\mathbb{E}\left [{\| \mathbf{\widetilde{r}}(t) \|} _{2}^{2} \right ] =\!\!
\begin{cases}
\frac{1}{2}  \sum\limits_{n=1}^{{N}} \left \lbrace \mathbf{s}^{H}_{\mathrm{E},n} {\mathbf{V}}_{\mathrm{E},n} \mathbf{s}_{\mathrm{E},n}
 + \sigma_{n}^{2} \textrm{tr}\left \{ \mathbf{U}_{\mathrm{E},n} \mathbf{U}^{H}_{\mathrm{E},n}  \right\} \right \rbrace, ~ \textrm{for~EH} , \\
\frac{1}{2}  \sum\limits_{n=1}^{{N}}   \left \lbrace \textrm{tr}\left \{  \mathbf{Q}_{\mathrm{I},n}  { \mathbf{V}}_{\mathrm{I},n} \right\} + \sigma_{n}^{2} \textrm{tr}\left \{ \mathbf{U}_{\mathrm{I},n} \mathbf{U}^{H}_{\mathrm{I},n} \right\} \right \rbrace,~ \textrm{for~ID},
\end{cases}}
\end{equation}
where $\sigma_{n}^{2}$ addresses both $\sigma_{\mathrm{R},n}^{2}$ and $\sigma_{\mathrm{IRS},n}^{2}$, $\mathbf{Q}_{\mathrm{I},n} = \textrm{Diag} (p_{\mathrm{I},1,n}, p_{\mathrm{I},2,n}, \cdots , p_{\mathrm{I},{K},n})$  and
\begin{equation} \label{keyVEtilde}
{ \mathbf{V}}_{\mathrm{E},n} = {\mathbf{H}}^{H}_{n} \mathbf{U}^{H}_{\mathrm{E},n} \mathbf{U}_{\mathrm{E},n} {\mathbf{H}}_n,  ~\forall n, ~~~ ~
{ \mathbf{V}}_{\mathrm{I},n} = {\mathbf{H}}^{H}_{n} \mathbf{U}^{H}_{\mathrm{I},n} \mathbf{U}_{\mathrm{I},n} {\mathbf{H}}_n,~ \forall n.
\end{equation}
Using~\eqref{newq1}, the total consumed energy is bounded at the relay/IRS in $t \in [0 , T]$ as
\begin{align} \label{newq}
\hspace{-12pt} \frac{\tau}{2\rho}   \big( \mathbf{s}^{H}_{\mathrm{E},n}   {\mathbf{V}}_{\mathrm{E},n}  \mathbf{s}_{\mathrm{E},n} +\sigma_{n}^{2} \textrm{tr}  \{ \mathbf{U}_{\mathrm{E},n} \mathbf{U}^{H}_{\mathrm{E},n}  \} \big)  + \frac{T- \tau}{2\rho}   \big( \textrm{tr}  \{  \mathbf{Q}_{\mathrm{I},n}  { \mathbf{V}}_{\mathrm{I},n}  \} +  \sigma_{n}^{2} \textrm{tr}  \{ \mathbf{U}_{\mathrm{I},n} \mathbf{U}^{H}_{\mathrm{I},n}  \} \big) \leq T  p^{\mathrm{rf}}_{n}, \forall n,
\end{align}
where $p^{\mathrm{rf}}_{n}$ denotes either the maximum power budget at the relay $p^{\mathrm{rf}}_{\mathrm{R},n}$ or IRS $p^{\mathrm{rf}}_{\mathrm{IRS}}$.
We can write received signal at $\mathcal{R}_k$ as
\begin{equation*} \label{bb} 
{{y}_{k} (t)=
\begin{cases}
\sum\limits_{n=1}^{{N}}  \left \lbrace  \mathbf{g}^{T}_{k,n} \mathbf{U}_{\mathrm{E},n} \left ( \mathbf{H}_n \mathbf{s}_{\mathrm{E},n} + \mathbf{z}_{n} \right)  + z_{k,n} \right \rbrace   ,~ t \in \widetilde{T}_{\mathrm{EH}},~ \forall k, ~ \textrm{for~EH} , \\
\sum\limits_{n=1}^{{N}} \left \lbrace    \mathbf{g}^{T}_{k,n}  \mathbf{U}_{\mathrm{I},n} \left ( \mathbf{H}_n \mathbf{s}_{\mathrm{I},n} + \mathbf{z}_{n}  \right)+ z_{k,n} + \widetilde{z}_{k,n} \right \rbrace  ,~  t \in \widetilde{T}_{\mathrm{ID}},~\forall k,~ \textrm{for~ID},
\end{cases}}
\end{equation*}
where $\mathbf{g}_{k,n}$ is the $k^{th}$ column vector of $\mathbf{G}^{T}_n$, and $z_{k,n}$ as well as $\widetilde{z}_{k,n}$ are the AWGN from the antenna  and baseband processing noises at $\mathcal{R}_{k}$, respectively, with $z_{k,n}\sim \mathcal{CN} (0, {\sigma}_{k,n}^{2})$ and $\widetilde{z}_{k,n} \sim \mathcal{CN} (0,{\delta}_{k,n}^{2})$.
The information signals at $\mathcal{R}_k$ corresponding to the $n^{th}$ subband can be expanded as
\begin{align} \label{bache}
{y}_{k,n} (t)=& \scalemath{1}{  {{\mathbf{g}}^{T}_{k,n} \mathbf{U}_{\mathrm{I},n} {\mathbf{h}}_{k,n} s_{\mathrm{I},k,n}}+ { {{\mathbf{g}}}^{T}_{k,n} \mathbf{U}_{\mathrm{I},n} \sum_{j=1,j\neq k}^{{K}} {\mathbf{h}}_{j,n} s_{\mathrm{I},j,n}  }}+    {\mathbf{g}}^{T}_{k,n} \mathbf{U}_{\mathrm{I},n} \mathbf{z}_{n}+z_{k,n}+\widetilde{z}_{k,n}, \hspace{.02cm} \forall k,n,
\end{align}
where ${\mathbf{h}}_{k,n}$ and ${\mathbf{g}}_{k,n}$ are the $k^{th}$ column vector of ${\mathbf{H}}_n$ and ${\mathbf{G}}^{T}_n$, respectively.
By defining $\mathbf{p}_{\mathrm{I},n}=[p_{\mathrm{I},1,n} , p_{\mathrm{I},2,n}, \cdots , p_{\mathrm{I},{K},n}]^{T}$, the SINR at the ID part for the $n^{th}$ subband is given by
\begin{align} \label{newqq}
&\scalemath{1}{\gamma_{k,n} (\mathbf{p}_{\mathrm{I},n}, \mathbf{U}_{\mathrm{I},n}) =\frac{ p_{\mathrm{I},k,n}  {\psi}_{k,k,n} }
{  \sum_{j=1,j \neq k}^{{K}} p_{\mathrm{I},j,n} {\psi}_{k,j,n}  +{\sigma}^2_{n}  {\widetilde{{\psi}}}_{k,n}   + {\delta}^2_{k,n}  +{\sigma}^2_{k,n}   },} \hspace{4pt} \forall k,n,
\end{align}
where
$
\psi_{k,j,n} =  {\mathbf{g}}^{T}_{k,n} \mathbf{U}_{\mathrm{I},n}  {\mathbf{h}}_{j,n} {\mathbf{h}}^{H}_{j,n} \mathbf{U}^{H}_{\mathrm{I},n}  {\mathbf{g}}^{*}_{k,n}$
and $
{\widetilde{{\psi}}}_{k,n} = {\mathbf{g}}^{T}_{k,n}  \mathbf{U}_{\mathrm{I},n} \mathbf{U}^{H}_{\mathrm{I},n}  {\mathbf{g}}^{*}_{k,n}$.
From Remark~\ref{kl}, we obtain the achievable rate at the $k^{th}$ pair as follows
\begin{align}\label{fg922}
\scalemath{1}{{R}_{k}\Big( \{ \mathbf{p}_{\mathrm{I},n} \} ^{{N}}_{n=1},\{ \mathbf{U}_{\mathrm{I},n}\} ^{{N}}_{n=1}, \tau \Big)= \frac{T-\tau}{\rho T} \sum_{n=1}^{{N}} \log_2 \Big(1 +   \gamma_{k,n} (\mathbf{p}_{\mathrm{I},n}, \mathbf{U}_{\mathrm{I},n}) \Big)}.
\end{align}
For the EH stream, we assume the noise power is negligible compared to the received signal power. We take into account the rectifier non-linearity by employing the results from~\cite{clerckx2018beneficial} where  the harvested energy at $\mathcal{R}_k$  is approximated by
\begin{equation}\label{fg9}
 {E}_{k}\Big( \{ \mathbf{s}_{\mathrm{E},n} \} ^{{N}}_{n=1},\{ \mathbf{U}_{\mathrm{E},n}\} ^{{N}}_{n=1}, \tau \Big)= \frac{\tau}{\rho}  \textrm{exp}\left( { \widetilde{a}    {\log}^2    p_{E,k} } \right) p^{  \widetilde{b}}_{E,k} \textrm{exp}{\widetilde{c}} ,~\forall k,
\end{equation}
where $\widetilde{a}$, $\widetilde{b}$, and $\widetilde{c}$ are the curve fitting constants and
$p_{\mathrm{E},k}$ is the average input power to $\mathcal{R}_k$'s harvester as
\begin{align} \label{newpow}
\scalemath{1}{p_{\mathrm{E},k} \Big( \{ \mathbf{s}_{\mathrm{E},n} \} ^{{N}}_{n=1},\{ \mathbf{U}_{\mathrm{E},n}\} ^{{N}}_{n=1}\Big) =\frac{1}{2} \sum_{n=1}^{{N}} \mathbf{s}^{H}_{\mathrm{E},n} {\hm{\Xi}}_{k,n} \mathbf{s}_{\mathrm{E},n},} \ \forall k,
\end{align}
with
\begin{align} \label{newpowmm}
{\hm{\Xi}}_{k,n} = {\mathbf{H}}^{H}_n \mathbf{U}^{H}_{\mathrm{E},n}  {\mathbf{g}}^{*}_{k,n} {\mathbf{g}}^{T}_{k,n}  \mathbf{U} _{\mathrm{E},n} {\mathbf{H}}_n, \  \forall k,n.
\end{align}
\begin{remark} \label{kl2}
Note that the reflection matrix cannot be designed separately for each subband in the IRS system, while, thanks to the RF chain circuits in a relay, the amplification matrix design is considered for each subband. We note that an active IRS is considerably less expensive than an AF relay. This is because an AF relay requires massive integrated circuits (including analog-to-digital/digital-to-analog converter, self-interference cancellation circuits, etc). The delay caused by RF chain processing of an AF relay contributes to latency, leads to lower transmission time, and requires more power for energy and information signals (see~\eqref{41} and~\eqref{newq}). Therefore, a relay-IRS trade-off exists in the system performance (see~\eqref{fg922} and~\eqref{fg9}).
\end{remark}
\begin{remark} \label{hh}
An approach with lower implementation complexity is considered in which only one amplification/reflection matrix needs to be designed for both energy and information time slots, called the t-static approach. Also, one can consider another approach with only one amplification matrix design in both time slots and all subbands, referred to as t-f-static in the relay system. These design methodologies lead to a lower signaling overhead and system performance.
\end{remark}
\section{The proposed Minimum Rate Maximization Method}\label{sum}
In this section, the aim is to maximize the minimum rate of the multi-user relay/IRS WPC system w.r.t. multi-sine energy waveforms $\mathbf{s}_{\mathrm{E},n}$, allocated power $\mathbf{p}_{\mathrm{I},n}$, amplification/reflection matrices $ \mathbf{U}_{\mathrm{E},n}, \mathbf{U}_{\mathrm{I},n}$, and the time allocation parameter $\tau$. The unified minimum rate maximization problem for both relay and IRS systems is cast as
\begin{align}\label{fg10}
	\max_{ \tau,  \lbrace  \mathbf{s}_{\mathrm{E},n} , \mathbf{p}_{\mathrm{I},n},   \mathbf{U}_{\mathrm{E},n} , \mathbf{U}_{\mathrm{I},n} \rbrace ^{{N}}_{n=1}  } &~ \displaystyle   \min_{1 \leq k \leq {\cal K}} ~~  {R}_{k} \\ \nonumber
	\mbox{s.t.}\;\;\;\;\;\;\;\;\;\;\;\;\;\;\;\;\;\;\;\;& \hspace{-50pt} \left \lbrace \tau,  \lbrace  \mathbf{s}_{\mathrm{E},n} , \mathbf{p}_{\mathrm{I},n},   \mathbf{U}_{\mathrm{E},n} , \mathbf{U}_{\mathrm{I},n} \rbrace ^{{N}}_{n=1} \right \rbrace \in \Omega,
\end{align}
where $\Omega = \Omega_0 \cap \Omega_{\mathrm{ind}}$ with
\begin{align}
\Omega_0=\Big\{& \textrm{C}_1:
 0 \leq \tau \leq T,  ~ \textrm{C}_2:\eqref{41}, \hspace{2pt}  p_{\mathrm{I},k,n}\geq 0, \hspace{1pt}  \forall k,n,~
   \textrm{C}_3:\eqref{newq},~ \textrm{C}_4: {E}_{k} \geq E_{\textrm{min},k} , \hspace{1pt}   \forall k \Big\},
\end{align}
\begin{eqnarray} \label{irs1}
\scalemath{1}{ \Omega_\mathrm{ind}= \begin{cases}
 \textrm{C}_{\mathrm{R}}: \mathbf{U}_{\mathrm{E},n},\mathbf{U}_{\mathrm{I},n} \in \complexC^{M_{\mathrm{R}} \times M_{\mathrm{R}}}, \  \forall n,  ~ \mbox{for relay},
 \\
  \textrm{C}_{\mathrm{IRS}}: \mathbf{U}_{\mathrm{E},n},\mathbf{U}_{\mathrm{I},n} \in \mathbb{D}^{M_{\mathrm{IRS}} \times M_{\mathrm{IRS}}}, \  \forall n, ~ \vert {\theta}_{\mathrm{E},m} \vert \geq 1, \   \vert {\theta}_{\mathrm{I},m} \vert \geq 1, \  \forall m,  ~ \textrm{for~IRS},
 \end{cases}}
\end{eqnarray}
and $E_{\textrm{min},k}$ in $\textrm{C}_4$ is the minimum required harvested energy for the $k^{th}$ user.

The problem in~\eqref{fg10} is non-convex due to the coupled design variables in the objective function and the constraints $\textrm{C}_2-\textrm{C}_4$ and $\textrm{C}_{\mathrm{IRS}}$.
To deal with this non-convex problem, we first solve the problem w.r.t. $\{ \mathbf{U}_{\mathrm{E},n}, \mathbf{U}_{\mathrm{I},n} \}$ for fixed $\{ \mathbf{s}_{\mathrm{E},n} , \mathbf{p}_{\mathrm{I},n}, \tau\}$, then optimize $\{ \mathbf{s}_{\mathrm{E},n}, \mathbf{p}_{\mathrm{I},n}\}$ for given $\{ \mathbf{U}_{\mathrm{E},n}, \mathbf{U}_{\mathrm{I},n}, \tau \}$, and finally, solve the problem w.r.t. $\tau$ via a closed-form solution. The procedure is repeated until convergence.
\subsection{Maximization over $\{ \mathbf{U}_{\mathrm{E},n} , \mathbf{U}_{\mathrm{I},n} \}$} \label{suma}
Here, we first consider the relay problem, and then the IRS problem is investigated.
\subsubsection{Relay System} \label{relaysub}
The problem in~\eqref{fg10} for fixed $\{\mathbf{s}_{\mathrm{E},n}, \mathbf{p}_{\mathrm{I},n}\}$ reduces to the following optimization
\begin{eqnarray} \label{neweq12}
	\max_{ {\lbrace \mathbf{U}_\mathrm{E}, \mathbf{U}_\mathrm{I} \rbrace}^{{N}}_{n=1}} & \displaystyle \min_{1 \leq k \leq {\cal K}} & \sum^{{N}}_{n=1}  {\log_2 \left(1 + \gamma_{k,n} (\mathbf{U}_{\mathrm{I},n})  \right)} \\ \nonumber
	\mbox{s.t.}\;\;&&
	\textrm{C}_3, ~\textrm{C}_4,
\end{eqnarray}
which is still a non-convex problem. To start solving the problem, first we need to reformulate the obtained expressions for the relay power constraint~\eqref{newq1}, SINR~\eqref{newqq}, and the input power of harvesters~\eqref{newpow} from Section~\ref{sys}. We can rewrite~\eqref{newq1} as (see Appendix \ref{app1} for the derivation)
\begin{equation} \label{bb2}
\scalemath{1}{ \mathbb{E} \left[{\| \mathbf{\widetilde{r}}(t) \|} _{2}^{2}\right] =
 \begin{cases}
 \frac{1}{2} \sum\limits_{n=1}^{{N}} \mathbf{u}^{H}_{\mathrm{E},n} \widetilde{\mathbf{A}}^{\mathrm{R}}_{\mathrm{E},n}   \mathbf{u}_{\mathrm{E},n}, ~ 0 \leq t \leq \tau,  ~ \textrm{for~EH}, \\
 \frac{1}{2} \sum\limits_{n=1}^{{N}}  \mathbf{u}^{H}_{\mathrm{I},n} \widetilde{\mathbf{A}}^{\mathrm{R}}_{\mathrm{I},n}   \mathbf{u}_{\mathrm{I},n}, ~\tau \leq t \leq T,  ~ \textrm{for~ID},
 \end{cases}}
\end{equation}
where $\mathbf{u}_{\mathrm{E},n} = \textrm{vec} ( \mathbf{U}_{\mathrm{E},n})$, $\mathbf{u}_{\mathrm{I},n} = \textrm{vec} ( \mathbf{U}_{\mathrm{I},n})$, and
\begin{equation} \label{nnh}
   \begin{array}{ll}
 \widetilde{\mathbf{A}}^{\mathrm{R}}_{\mathrm{E},n} =   {\left( {\mathbf{H}}_n \mathbf{s}_{\mathrm{E},n} \mathbf{s}^{H}_{\mathrm{E},n} {\mathbf{H}}^{H}_n \right)}^{T} \otimes \mathbf{I}_{{M_{\mathrm{R}}}} +\sigma_{n}^{2} \mathbf{I}_{M_{\mathrm{R}}^{2}},
~~~
 \widetilde{\mathbf{A}}^{\mathrm{R}}_{\mathrm{I},n} =   {\left( {\mathbf{H}}_n \mathbf{Q}_{\mathrm{I},n} {\mathbf{H}}^{H}_n \right)}^{T} \otimes \mathbf{I}_{{M_{\mathrm{R}}}} +\sigma_{n}^{2} \mathbf{I}_{M_{\mathrm{R}}^{2}}.  
\end{array}
\end{equation}
Therefore, we rewrite the relay power constraint in \eqref{newq} by using \eqref{bb2} as
\begin{equation} \label{bb3}
 \frac{\tau}{2\rho}  \mathbf{u}^{H}_{\mathrm{E},n}
 \widetilde{\mathbf{A}}^{\mathrm{R}}_{\mathrm{E},n}  \mathbf{u}_{\mathrm{E},n} + \frac{T- \tau}{2 \rho}  \mathbf{u}^{H}_{\mathrm{I},n}
 \widetilde{\mathbf{A}}^{\mathrm{R}}_{\mathrm{I},n}  \mathbf{u}_{\mathrm{I},n} \leq T p^{\mathrm{rf}}_{{\mathrm{R}},n},\hspace{5pt} \forall n.
\end{equation}
Next, we rewrite the SINR and the input power at $\mathcal{R}_k$'s harvester in~\eqref{newqq} and~\eqref{newpow} as
\begin{align} \label{newqq1}
 \gamma_{k,n} (\mathbf{u}_{\mathrm{I},n})
 =\frac{\mathbf{u}^{H}_{\mathrm{I},n} {\mathbf{A}}^{\mathrm{R}}_{k,n} \mathbf{u}_{\mathrm{I},n}}
 { \mathbf{u}^{H}_{\mathrm{I},n} \widehat{\mathbf{A}}^{\mathrm{R}}_{k,n} \mathbf{u}_{\mathrm{I},n} +{\delta}^2_{k,n}  +{\sigma}^2_{k,n}   },\hspace{5pt} \forall k,n,
\\ \label{fg9na}
 p_{\mathrm{E},k} \Big(\{ \mathbf{u}_{\mathrm{E},n}\} ^{{N}}_{n=1}\Big)= \frac{1}{2} \sum_{n=1}^{{N}} \mathbf{u}^{H}_{\mathrm{E},n} \bar{\mathbf{A}}^{\mathrm{R}}_{k,n} \mathbf{u}_{\mathrm{E},n},\hspace{5pt} \forall k,
\end{align}
where
\begin{align} \label{keyA}
 {\mathbf{A}}^{\mathrm{R}}_{k,n} &=   p_{\mathrm{I},k,n}   \left( {\mathbf{h}}_{k,n} {\mathbf{h}}^{H}_{k,n} \right)^{T} \otimes {\mathbf{g}}^{*}_{k,n} {\mathbf{g}}^{T}_{k,n}   ,
\\ \label{keyAhat}
 \widehat{\mathbf{A}}^{\mathrm{R}}_{k,n}  &\scalemath{1}{= \sum_{j=1,j \neq k}^{{K}}  p_{\mathrm{I},j,n}     \left(  {\mathbf{h}}_{j,n} {\mathbf{h}}^{H}_{j,n} \right)^T  \otimes  {\mathbf{g}}^{*}_{k,n} {\mathbf{g}}^{T}_{k,n} + {\sigma}^2_{n}    \mathbf{I}_{{M_{\mathrm{R}}}} \otimes  {\mathbf{g}}^{*}_{k,n} {\mathbf{g}}^{T}_{k,n}      ,}
\\ \label{keyAbar}
 \bar{\mathbf{A}}^{\mathrm{R}}_{k,n} &=  {\left( {\mathbf{H}}_n \mathbf{s}_{\mathrm{E},n}  \mathbf{s}^{H}_{\mathrm{E},n} {\mathbf{H}}^{H}_n    \right)}^{T} \otimes   {\mathbf{g}}^{*}_{k,n} {\mathbf{g}}^{T}_{k,n}.
\end{align}
By using \eqref{bb3}, \eqref{newqq1}, and \eqref{fg9na} with an auxiliary variable $\alpha_{a}$ the optimization problem in \eqref{neweq12} can be equivalently rewritten as
\begin{align} \label{maxmin4}
& \max_{\alpha_{a},{\{ \mathbf{u}_\mathrm{E}, \mathbf{u}_\mathrm{I} \}}^{{N}}_{n=1}}  \alpha_{a} \\ \nonumber
 & \mbox{s.t.}~~~ \textrm{C}_3:\eqref{bb3},~\textrm{C}_4:{{E}}_{k}\left( \{ \mathbf{u}_{\mathrm{E},n}\} ^{{N}}_{n=1}\right) \geq {E_{\textrm{min},k}} ,  \hspace{2pt} \forall k,
  \\ \nonumber&  ~~~~~~\hspace{2pt}\textrm{C}_5:  \sum^{{N}}_{n=1} {\log_2 \left(1 + \frac{ \mathbf{u}^{H}_{\mathrm{I},n} {\mathbf{A}}^{\mathrm{R}}_{k,n} \mathbf{u}_{\mathrm{I},n}  }
   {\mathbf{u}^{H}_{\mathrm{I},n} \widehat{\mathbf{A}}^{\mathrm{R}}_{k,n} \mathbf{u}_{\mathrm{I},n} +  {\zeta}_{k,n,a}} \right)}   \geq \alpha_{a},  \hspace{2pt} \forall k,
\end{align}
where ${\zeta}_{k,n,a} = {\sigma}^2_{k,n}+{\delta}^2_{k,n}$.
The constraint $\textrm{C}_5$ can be equivalently rewritten as
\begin{align} \label{fg101n}
	\textrm{C}_5:\sum^{{N}}_{n=1}  \Big \lbrace \log_2 \left(\mathbf{u}^H_{\mathrm{I},n} \mathbf{B}_{k,n} \mathbf{u}_{\mathrm{I},n} +{\zeta}_{k,n,a} \right) - \log_2 \left( \mathbf{u}^H_{\mathrm{I},n} \widehat{\mathbf{A}}^{\mathrm{R}}_{k,n} \mathbf{u}_{\mathrm{I},n}+{\zeta}_{k,n,a} \right)  \Big \rbrace \geq {\alpha_{a}},
\end{align}
where $\mathbf{B}_{k,n} = \widehat{\mathbf{A}}^{\mathrm{R}}_{k,n} +  {\mathbf{A}}^{\mathrm{R}}_{k,n} $. It is observed that this constraint is non-convex. Therefore, we employ the MM technique to tackle its non-convexity.
Precisely,  we minorize  the denominator term $\scalemath{1} {- \log_2 \Big( \mathbf{u}^H_{\mathrm{I},n} \widehat{\mathbf{A}}^{\mathrm{R}}_{k,n} \mathbf{u}_{\mathrm{I},n}}$ $\scalemath{.93} {+{\zeta}_{k,n,b} \Big) }$  by the using the following inequality
\begin{equation}\label{fg103}
 \log_2 (x) \leq \log_2 (x_0) + \frac{\log_2 e}{x_0} (x - x_0).
\end{equation}
By setting $x=  \mathbf{u}^H_{\mathrm{I},n} \widehat{\mathbf{A}}^{\mathrm{R}}_{k,n} \mathbf{u}_{\mathrm{I},n}+{\zeta}_{k,n,a} $ and $ x_0=  {\left(\mathbf{u}^{0}_{\mathrm{I},n}\right)}^{H} \widehat{\mathbf{A}}^{\mathrm{R}}_{k,n} \mathbf{u}^0_{\mathrm{I},n} +{\zeta}_{k,n,a}$ in \eqref{fg103} we obtain
\begin{align}
 - \log_2 \left( \mathbf{u}^H_{\mathrm{I},n} \widehat{\mathbf{A}}^{\mathrm{R}}_{k,n} \mathbf{u}_{\mathrm{I},n} +{\zeta}_{k,n,a} \right)   \geq & - \log_2 \left( \left(\mathbf{u}^{0}_{\mathrm{I},n}\right)^H \widehat{\mathbf{A}}^{\mathrm{R}}_{k,n} \mathbf{u}^0_{\mathrm{I},n} +{\zeta}_{k,n,a} \right) \\ \nonumber &- \frac{\log_2 e  \hspace{5pt} \Big(  \mathbf{u}^H_{\mathrm{I},n} \widehat{\mathbf{A}}^{\mathrm{R}}_{k,n} \mathbf{u}_{\mathrm{I},n}  - \left( \mathbf{u}^{0}_{\mathrm{I},n}\right)^H \widehat{\mathbf{A}}^{\mathrm{R}}_{k,n} \mathbf{u}^0_{\mathrm{I},n} \Big)}{ \left( \mathbf{u}^{0}_{\mathrm{I},n}\right)^H \widehat{\mathbf{A}}^{\mathrm{R}}_{k,n} \mathbf{u}^0_{\mathrm{I},n} +{\zeta}_{k,n,a}}  .
\end{align}
Applying the above minorizer, the constraint $\textrm{C}_5$ in~\eqref{fg101n} is rewritten
at the $ i^{th}$ iteration of the MM technique as
\begin{align}\label{maxmin4n}
	\sum^{{N}}_{n=1} \bigg \lbrace& \log_2 \left(\mathbf{u}^H_{\mathrm{I},n} \mathbf{B}_{k,n} \mathbf{u}_{\mathrm{I},n} +{\zeta}_{k,n,a} \right) -
	\log_2 \left( \left(\mathbf{u}^{(i-1)}_{\mathrm{I},n}\right)^{H} \widehat{\mathbf{A}}^{\mathrm{R}}_{k,n} \mathbf{u}^{(i-1)}_{\mathrm{I},n}+{\zeta}_{k,n,a} \right)  \\ \nonumber & - \frac{\log_2 e}{\left(\mathbf{u}^{(i-1)}_{\mathrm{I},n}\right)^{H} \widehat{\mathbf{A}}^{\mathrm{R}}_{k,n} \mathbf{u}^{(i-1)}_{\mathrm{I},n} +{\zeta}_{k,n,a}} \hspace{2pt}  \left(  \mathbf{u}^H_{\mathrm{I},n} \widehat{\mathbf{A}}^{\mathrm{R}}_{k,n} \mathbf{u}_{\mathrm{I},n} - \left(\mathbf{u}^{(i-1)}_{\mathrm{I},n}\right)^{H} \widehat{\mathbf{A}}^{\mathrm{R}}_{k,n} \mathbf{u}^{(i-1)}_{\mathrm{I},n} \right) \bigg \rbrace
	\geq {\alpha_a}.
\end{align}
The following lemma lays the ground for dealing with the first non-concave logarithmic term in~\eqref{maxmin4n} in light of the MM technique.
\begin{lemma} \label{fmn}
 Let $s(\mathbf{x})= - \log_2 \left( \mathbf{x}^H \mathbf{T} \mathbf{x} + \nu \right)$ and $\mathbf{x}^H \mathbf{Q} \mathbf{x} \leq P$ for any positive-definite matrices $ \mathbf{T}, \mathbf{Q}\in \mathbb{S}_{++}^{N}$ and $P \in \realR_{+}$. Then, $s(\mathbf{x})$ is bounded for all $\mathbf{x} $ and $ \mathbf{x}_0$ as follows
 \begin{align*}
 s(\mathbf{x}) \leq \hspace{2pt}& s(\mathbf{x}_0) + \Re \left \{ \mathbf{b}^H (\mathbf{x} - \mathbf{x}_0) \right \}
 + (\mathbf{x} - \mathbf{x}_0)^H  \mathbf{D} (\mathbf{x} - \mathbf{x}_0),
 \end{align*}
where $\mathbf{b}=  \frac{-2\log_2 e}{\mathbf{x}^{H}_{0} \mathbf{T} \mathbf{x}_0  + \nu}   \mathbf{T} \mathbf{x}_0$, $\mathbf{D}=  \frac{4P}{\mathbf{w}^{H}_{1} \mathbf{Q} \mathbf{w}_1 }   \mathbf{I}_{{{M^2_{\mathrm{R}}}}}$,
 and $\mathbf{w}_1$ is the principal eigenvector of $\mathbf{T}$ and $\epsilon > 0$.
\end{lemma}
\begin{proof}See Appendix \ref{app4}.\end{proof}
 Using Lemma~\ref{fmn} and noting that the term $\frac{\tau}{2}  \mathbf{u}^{H}_{\mathrm{E},n} \widetilde{\mathbf{A}}^{\mathrm{R}}_{\mathrm{E},n}  \mathbf{u}_{\mathrm{E},n}$ in~\eqref{bb3} is positive, we obtain the following minorizer for the term $\log_2 (\mathbf{u}^H_{\mathrm{I},n} \mathbf{B}_{k,n} \mathbf{u}_{\mathrm{I},n} +{\zeta}_{k,n,a})$ in~\eqref{maxmin4n} at any given $\mathbf{u}^0_{\mathrm{I},n}$
\begin{align}\label{fg1o09}
 \log_2 (\mathbf{u}^H_{\mathrm{I},n} \mathbf{B}_{k,n} \mathbf{u}_{\mathrm{I},n} +{\zeta}_{k,n,a})  \geq & \log_2 \left( \left(\mathbf{u}^{0}_{\mathrm{I},n}\right)^H \mathbf{B}_{k,n} \mathbf{u}^0_{\mathrm{I},n}  +{\zeta}_{k,n,a} \right) - \Re \left \{ \mathbf{b}^{H}_{k,n} (\mathbf{u}_{\mathrm{I},n} - \mathbf{u}^0_{\mathrm{I},n})\right \}  \\ \nonumber & - \left(\mathbf{u}_{\mathrm{I},n} - \mathbf{u}^0_{\mathrm{I},n} \right)^H \mathbf{D}_{k,n} (\mathbf{u}_{\mathrm{I},n} - \mathbf{u}^0_{\mathrm{I},n}),
\end{align}
where
\begin{equation*}
	\mathbf{b}_{k,n}=  \frac{-2\log_2 e}{\left(\mathbf{u}^{0}_{\mathrm{I},n}\right) ^{H} \mathbf{B}_{k,n} \mathbf{u}^{0}_{\mathrm{I},n} +{\zeta}_{k,n,a}}   \mathbf{B}_{k,n} \mathbf{u}^{0}_{\mathrm{I},n},
	\hspace{30pt}
	\mathbf{D}_{k,n}=  \frac{\frac{16T}{T-\tau}p^{\mathrm{rf}}_{{\mathrm{R}},n}}{\widetilde{\mathbf{w}}^{H}_{k,n} \widetilde{\mathbf{A}}^{\mathrm{R}}_{\mathrm{I},n} \widetilde{\mathbf{w}}_{k,n} }   \mathbf{I}_{M^2_{\mathrm{R}}},
\end{equation*}
and $\widetilde{\mathbf{w}}_{k,n}$ denotes the principal eigenvector of $\mathbf{B}_{k,n}$.
Applying~\eqref{fg1o09}, the constraint in~\eqref{maxmin4n} is restated as
\begin{align}\label{maxmin5n}
\hspace{-5pt} \scalemath{.93}{- \sum^{{N}}_{n=1} \bigg \{   \frac{\log_2 e \  \mathbf{u}^H_{\mathrm{I},n} \widehat{\mathbf{A}}^{\mathrm{R}}_{k,n} \mathbf{u}_{\mathrm{I},n}}{\left(\mathbf{u}^{(i -1)}_{\mathrm{I},n}\right) ^{H} \widehat{\mathbf{A}}^{\mathrm{R}}_{k,n} \mathbf{u}^{(i -1)}_{\mathrm{I},n} +{\zeta}_{k,n,a}}  + \mathbf{u}^H_{\mathrm{I},n} \mathbf{D}_{k,n} \mathbf{u}_{\mathrm{I},n} +\Re \left \{       \left( \mathbf{b}_{k,n} -2\mathbf{D}_{k,n} \mathbf{u}^{(i -1)}_{\mathrm{I},n} \right)^H \mathbf{u}_{\mathrm{I},n} \right \}  + d^{(i )}_{k,n} \bigg \} \geq {\alpha_a}, \hspace{0pt}  \forall k,}
\end{align}
where
\begin{align}\label{maxmin7}
 d^{(i)}_{k,n} =\hspace{2pt}&\log_2  \frac{\left(\mathbf{u}^{(i -1)}_{\mathrm{I},n}\right)^H \widehat{\mathbf{A}}^{\mathrm{R}}_{k,n} \mathbf{u}^{(i -1)}_{\mathrm{I},n} +{\zeta}_{k,n,a}}{\left(\mathbf{u}^{(i -1)}_{\mathrm{I},n}\right)^H \mathbf{B}_{k,n} \mathbf{u}^{(i -1)}_{\mathrm{I},n}+{\zeta}_{k,n,a}}   - \Re \left \{  \mathbf{b}^{H}_{k,n} \mathbf{u}^{(i -1)}_{\mathrm{I},n} \right \}
 +\left(\mathbf{u}^{(i -1)}_{\mathrm{I},n}\right)^H \mathbf{D}_{k,n} \mathbf{u}^{(i -1)}_{\mathrm{I},n}\\ \nonumber & - \frac{\log_2 e \  \left(  \mathbf{u}^{(i-1)}_{\mathrm{I},n}\right)^{H} \widehat{\mathbf{A}}^{\mathrm{R}}_{k,n} \mathbf{u}^{(i-1)}_{\mathrm{I},n} }{\left(\mathbf{u}^{(i-1)}_{\mathrm{I},n}\right)^{H} \widehat{\mathbf{A}}^{\mathrm{R}}_{k,n} \mathbf{u}^{(i-1)}_{\mathrm{I},n} +{\zeta}_{k,n,a}}.
\end{align}
Then, we can simplify constraint in~\eqref{maxmin5n} as
\begin{align}\label{maxmin6}
 -\sum^{{N}}_{n=1} \Big \{& \mathbf{u}^H_{\mathrm{I},n} {\mathbf{F}}^{(i )}_{k,n} \mathbf{u}_{\mathrm{I},n} +\Re \left \{           ( {\mathbf{f}}^{(i )}_{k,n})^H \mathbf{u}_{\mathrm{I},n} \right \}+d^{(i)}_{k,n} \Big \} \geq {{\alpha_a}}, \   \forall k,
\end{align}
where 
\begin{equation}\label{max1F}
 {\mathbf{F}}^{(i)}_{k,n}= \frac{\log_2 e \hspace{4pt} \widehat{\mathbf{A}}^{\mathrm{R}}_{k,n}}{\left(\mathbf{u}^{(i -1)}_{\mathrm{I},n}\right) ^{H} \widehat{\mathbf{A}}^{\mathrm{R}}_{k,n} \mathbf{u}^{(i-1 )}_{\mathrm{I},n} +{\zeta}_{k,n,a}} + \mathbf{D}_{k,n},~~~~ {\mathbf{f}}^{(i)}_{k,n}=  \mathbf{b}_{k,n} - 2  \mathbf{D}_{k,n} \mathbf{u}^{(i -1)}_{\mathrm{I},n} .
\end{equation}
Finally, we focus on the constraint $\textrm{C}_4$. From \eqref{fg9} and \eqref{fg9na}, we see that in the left-hand side (LHS) of $\textrm{C}_4$, ${{E}}_{k}$ is neither convex nor concave w.r.t. $\mathbf{u}_{\mathrm{E},n}$.
To apply the MM technique on LHS of $\textrm{C}_4$, we first define a parameter\footnote{See Appendix~\ref{beta12} for a selection of ${\beta}_{k,n,a}$.} ${\beta}_{k,n,a}$ such that ${\nabla}^{2}_{{\mathbf{u}}_{\mathrm{E},n}} {{E}}_{k} \left( \lbrace \mathbf{u}_{\mathrm{E},n}\rbrace ^{{N}}_{n=1}\right) + {\beta}_{k,n,a} \mathbf{I}_{M^2_{\mathrm{R}}} \succeq \mathbf{0}, \hspace{3pt} \forall k,n,$ and write ${{E}}_{k}$ as the sum of a convex and a concave function as
\begin{align}\label{jj}
 {{E}}_{k} \left( \{ \mathbf{u}_{\mathrm{E},n}\} ^{{N}}_{n=1}\right) =& {{E}}_{k} \left( \{ \mathbf{u}_{\mathrm{E},n}\} ^{{N}}_{n=1}\right)  + \frac{1}{2}   \sum_{n=1}^{{N}}  {\beta}_{k,n,a} {\mathbf{u}^{{H}}_{\mathrm{E},n}} \mathbf{u}_{\mathrm{E},n}
- \frac{1}{2}  \sum_{n=1}^{{N}} \beta_{k,n,a} {\mathbf{u}^{{H}}_{\mathrm{E},n}} \mathbf{u}_{\mathrm{E},n},~\forall k.
\end{align}
We now apply the MM technique to $\textrm{C}_4$ and obtain a convex constraint. To do so, we keep the concave part and minorize the convex part of~\eqref{jj} and rewrite  $\textrm{C}_4$ as
\begin{align}\label{c4b}
 & {{E}}_{k} \left( \{ \mathbf{u}^{(i-1)}_{\mathrm{E},n}\} ^{{N}}_{n=1}\right)  + \frac{1}{2}   \sum_{n=1}^{{N}}   {{\beta}}_{k,n,a} {\left( \mathbf{u}^{{(i-1)}}_{\mathrm{E},n}\right)}^{H} \mathbf{u}^{(i-1)}_{\mathrm{E},n} + \sum_{n=1}^{{N}}    \Re  \left \{ {{\hm{\vartheta}}}^{(i)}_{k,n,a}  \left(\mathbf{u}_{\mathrm{E},n} - {\mathbf{u}}^{(i -1)}_{\mathrm{E},n} \right) \right \} \\ \nonumber &  - \frac{1}{2}   \sum_{n=1}^{{N}}  {{\beta}}_{k,n,a} {\mathbf{u}^{H}_{\mathrm{E},n}} \mathbf{u}_{\mathrm{E},n} \geq {E_{\textrm{min},k}}  ,  \hspace{5pt} \forall k,
\end{align}
where 
\begin{align*}
 {{\hm{\vartheta}}}^{(i)}_{k,n,a}  =& {\beta}_{k,n,a}   \left({\mathbf{u}_{\mathrm{E},n}^{(i-1)}}\right)^{{H}}  +  \frac{\tau \textrm{exp}{\widetilde{c}}}{2}\textrm{exp} \left({ \widetilde{a}  { {\log}^2  {\omega}_{k,a}^{(i)}    }  }\right) {\left(   {\omega}_{k,a}^{(i)}  \right) }^{\widetilde{b}-1}   \left({2\widetilde{a}  \log    {\omega}_{k,a}^{(i)} }  \hspace{2pt}    +\widetilde{b} \right)  \left({{\mathbf{u}}^{(i-1)}_{\mathrm{E},n}}\right)^{{H}} {\bar{\mathbf{A}}}^{\mathrm{R}}_{k,n},
\end{align*}
with ${\omega}_{k,a}^{(i)} = \frac{1}{2} \sum_{n=1}^{{N}} \left({{\mathbf{u}}^{(i-1)}_{\mathrm{E},n}}\right)^{{H}} {\bar{\mathbf{A}}}^{\mathrm{R}}_{k,n} {\mathbf{u}}^{(i-1)}_{\mathrm{E},n}$. Therefore, the $i^{th}$ MM iteration for~\eqref{neweq12} is the solution of the following convex problem 
\begin{align} \label{kkk12}
 &\max_{\alpha_a, \{ \mathbf{u}_{\mathrm{E},n}, \mathbf{u}_{\mathrm{I},n} \} ^{{N}}_{n=1}}   \displaystyle \alpha_a \\ \nonumber &
 \mbox{s.t.}~~ \textrm{C}_3:{\eqref{bb3}},~ \textrm{C}_4:\eqref{c4b}, ~ \textrm{C}_5:\eqref{maxmin6},
\end{align}
which can be solved efficiently.
\subsubsection{IRS System} \label{irssub}
By considering $\mathbf{U}_{\mathrm{E},n}=\textrm{Diag}({\hm{\theta}_\mathrm{E}})$, $\mathbf{U}_{\mathrm{I},n}=\textrm{Diag}({\hm{\theta}_\mathrm{I}})$ and adding the constraint $\textrm{C}_{\mathrm{IRS}}$ in~\eqref{irs1}, the optimization problem in~\eqref{neweq12} is considered in this subsection. Since $\mathbf{U}_{\mathrm{E},n}$ and $\mathbf{U}_{\mathrm{I},n}$ are diagonal matrices, the expressions in \eqref{bb3}-\eqref{fg9na} are modified as
\begin{equation} \label{bb31}
 \frac{\tau}{2}  \hm{\theta}^{H}_\mathrm{E}
 \widetilde{\mathbf{A}}^{\mathrm{IRS}}_{\mathrm{E},n}  \hm{\theta}_{E} + \frac{T- \tau}{2}  \hm{\theta}^{H}_\mathrm{I}
 \widetilde{\mathbf{A}}^{\mathrm{IRS}}_{\mathrm{I},n}  \hm{\theta}_\mathrm{I} \leq T p^{\mathrm{rf}}_{\mathrm{IRS}}, \hspace{.05cm} \forall n,
\end{equation}
\begin{align} \label{newqq11}
 \gamma_{k,n} (\hm{\theta}_\mathrm{I})
 =\frac{\hm{\theta}^{H}_\mathrm{I} {\mathbf{A}}^{\mathrm{IRS}}_{k,n} \hm{\theta}_\mathrm{I}}
 { \hm{\theta}^{H}_\mathrm{I} \widehat{\mathbf{A}}^{\mathrm{IRS}}_{k,n} \hm{\theta}_\mathrm{I} +{\delta}^2_{k,n}  +{\sigma}^2_{k,n}   },\hspace{5pt} \forall k,n,
\end{align}
\begin{align} \label{fg9na1}
 p_{E,k} \left(\hm{\theta}_\mathrm{E}\right)= \frac{1}{2} \sum_{n=1}^{{N}} \hm{\theta}^{H}_\mathrm{E} \bar{\mathbf{A}}^{\mathrm{IRS}}_{k,n} \hm{\theta}_\mathrm{E},\hspace{5pt} \forall k,
\end{align}
where their parameters are defined in Lemma~2 below.
\begin{lemma}
 The parameters $ \widetilde{\mathbf{A}}^{\mathrm{IRS}}_{\mathrm{E},n}$, $ \widetilde{\mathbf{A}}^{\mathrm{IRS}}_{\mathrm{I},n}$, ${\mathbf{A}}^{\mathrm{IRS}}_{k,n}$, $\widehat{\mathbf{A}}^{\mathrm{IRS}}_{k,n}$, and $\bar{\mathbf{A}}^{\mathrm{IRS}}_{k,n}$ are expressed as follows:
\begin{align} \label{nnh1}
 \widetilde{\mathbf{A}}^{\mathrm{IRS}}_{\mathrm{E},n} &=   {\left( {\mathbf{H}}_n \mathbf{s}_{\mathrm{E},n} \mathbf{s}^{H}_{\mathrm{E},n} {\mathbf{H}}^{H}_n \right)}^{T} \odot \mathbf{I}_{{M_{\mathrm{IRS}}}}  +\sigma_{n}^{2} \mathbf{I}_{{M_{\mathrm{IRS}}}},
  \\
 \widetilde{\mathbf{A}}^{\mathrm{IRS}}_{\mathrm{I},n} &=   {\left( {\mathbf{H}}_n \mathbf{Q}_{\mathrm{I},n} {\mathbf{H}}^{H}_n \right)}^{T} \odot \mathbf{I}_{{M_{\mathrm{IRS}}}}  +\sigma_{n}^{2} \mathbf{I}_{{M_{\mathrm{IRS}}}},
\\ \label{keyA1}
 {\mathbf{A}}^{\mathrm{IRS}}_{k,n} &=   p_{\mathrm{I},k,n}   \left( {\mathbf{h}}_{k,n} {\mathbf{h}}^{H}_{k,n} \right)^{T} \odot {\mathbf{g}}^{*}_{k,n} {\mathbf{g}}^{T}_{k,n}   ,
\\ \label{keyAhat1}
 \widehat{\mathbf{A}}^{\mathrm{IRS}}_{k,n}  &=\scalemath{1}{\sum_{j=1,j \neq k}^{{K}} p_{\mathrm{I},j,n}     \left(  {\mathbf{h}}_{j,n} {\mathbf{h}}^{H}_{j,n} \right)^{T}\!\!\!  \odot  {\mathbf{g}}^{*}_{k,n} {\mathbf{g}}^{T}_{k,n} \!\!+ {\sigma}^2_{n}    \mathbf{I}_{{M_{\mathrm{IRS}}}} \!  \odot {\mathbf{g}}^{*}_{k,n} {\mathbf{g}}^{T}_{k,n}      ,}
\\ \label{keyAbar1}
 \bar{\mathbf{A}}^{\mathrm{IRS}}_{k,n} &=  {\left( {\mathbf{H}}_n \mathbf{s}_{\mathrm{E},n}  \mathbf{s}^{H}_{\mathrm{E},n} {\mathbf{H}}^{H}_n    \right)}^{T} \odot   {\mathbf{g}}^{*}_{k,n} {\mathbf{g}}^{T}_{k,n}.
\end{align}
It is worth pointing out that the only difference between the parameters above and their corresponding expressions in~\eqref{nnh} and~\eqref{keyA}-\eqref{keyAbar}, is the symbol of multiplication, i.e., $\otimes$ and $\odot$, in a proper dimension. The proper dimension consideration means ${{M_{\mathrm{R}}}} \rightarrow {{M_{\mathrm{IRS}}}}$ for all of the above parameters and $\mathbf{I}_{{M^2_{\mathrm{R}}}} \rightarrow \mathbf{I}_{{M_{\mathrm{IRS}}}}$ for the second terms of $\widetilde{\mathbf{A}}^{\mathrm{IRS}}_{\mathrm{E},n}$ and $ \widetilde{\mathbf{A}}^{\mathrm{IRS}}_{\mathrm{I},n}$.
\end{lemma}
\begin{proof} See Appendix \ref{app1000}.\end{proof}
Next, we focus on constraint $\textrm{C}_{\mathrm{IRS}}$. First, let us introduce the following minorizer~\cite{song2015sequence}
\begin{equation}
 \vert x \vert \geq \Re \left \{ x^{*} \frac{x_0}{\vert x_0 \vert}  \right \}.
\end{equation}
Then, considering the above minorizer, the constraint $\textrm{C}_{\mathrm{IRS}}$ is expressed as the $i^{th}$ iteration of MM as
\begin{equation} \label{nnb}
 \Re \left \{ {\theta}^{*}_{\mathrm{E},m} \frac{{\theta}^{(i-1)}_{\mathrm{E},m}}{\vert {\theta}^{(i-1)}_{\mathrm{E},m} \vert}  \right \} \geq 1, \  \Re \left \{ {\theta}^{*}_{\mathrm{I},m} \frac{{\theta}^{(i-1)}_{\mathrm{I},m}}{\vert {\theta}^{(i-1)}_{\mathrm{I},m} \vert}  \right \} \geq 1, \hspace{8pt} \forall m.
\end{equation}
Therefore, the optimization problem in~\eqref{maxmin4} is modified as
\begin{eqnarray} \label{maxmin41}
 \max_{\alpha_{a}, \hm{\theta}_{\mathrm{E}}, \hm{\theta}_\mathrm{I}}  && \alpha_{a} \\ \nonumber
 \mbox{s.t.}\;\;&& \textrm{C}_3:\eqref{bb31},~\textrm{C}_4:{{E}}_{k}\left( \hm{\theta}_\mathrm{E}\right) \geq {E_{\textrm{min},k}} ,  \hspace{2pt} \forall k,~\textrm{C}_{\mathrm{IRS}}:~\eqref{nnb},
 \\ \nonumber&& \textrm{C}_5:  \sum^{{N}}_{n=1} {\log_2 \left(1 + \frac{ \hm{\theta}^{H}_\mathrm{I} {\mathbf{A}}^{\mathrm{IRS}}_{k,n} \hm{\theta}_\mathrm{I}  }
 {\hm{\theta}^{H}_\mathrm{I} \widehat{\mathbf{A}}^{\mathrm{IRS}}_{k,n} \hm{\theta}_\mathrm{I} +  {\zeta}_{k,n,a}} \right)}   \geq \alpha_{a},  \hspace{2pt} \forall k,
\end{eqnarray}
where the steps for constraints $\textrm{C}_3$-$\textrm{C}_5$ in Subsection~\ref{relaysub} are used exactly here.
\subsection{Maximization over $\{ \mathbf{s}_{\mathrm{E},n}, \mathbf{p}_{\mathrm{I},n}\}$} \label{mina}
By introducing an auxiliary variable $\alpha_b$, the relay/IRS problem in~\eqref{fg10} for fixed $\{\mathbf{U}_{\mathrm{E},n}, \mathbf{U}_{\mathrm{I},n} , \tau\}$ boils down to the following optimization:
\begin{eqnarray}\label{fg105newa}
 \max_{\alpha_b, \{ \mathbf{s}_{\mathrm{E},n}, \mathbf{p}_{\mathrm{I},n} \} ^{{N}}_{n=1}}  & \alpha_b  \\ \nonumber
 \mbox{s.t.}\;\;&& \hspace{-5pt} \textrm{C}_2:\eqref{41}, \hspace{3pt} p_{\mathrm{I},k,n}\geq 0,  \forall k,n, \hspace{7pt}
 \textrm{C}_3:\eqref{newq},  \hspace{7pt} \textrm{C}_4: {E}_{k} \left(\{ \mathbf{s}_{\mathrm{E}}\}_{n=1}^N \right) \geq E_{\textrm{min},k} ,   \forall k,
 \\ \nonumber &&
 \textrm{C}_5:
 \sum_{n=1}^{{N}} {\log_2 \left(1 + \gamma_{k,n} (\mathbf{p}_{\mathrm{I},n}) \right)} \geq \alpha_b,   \forall k.
\end{eqnarray}
The constraints $\textrm{C}_4$ and $\textrm{C}_5$ of this sub-problem are non-convex. We first rewrite the SINR associated with the $k^{th}$ pair in~\eqref{newqq} as
\begin{equation} \label{tr}
 \gamma_{k,n} (\mathbf{p}_{\mathrm{I},n}) = \frac{ { \mathbf{a}}_{k,n}^{T}\mathbf{p}_{\mathrm{I},n} }{{ \mathbf{b}}_{k,n}^{T}\mathbf{p}_{\mathrm{I},n} + {\sigma}^2_{n}  {\widetilde{{\psi}}}_{k,n}  + {\delta}^2_{k,n} +{\sigma}^2_{k,n} },
\end{equation}
where ${\mathbf{a}}_{k,n}= \psi_{k,k,n} \mathbf{e}_k$, ${\mathbf{b}}_{k,n}= [{\psi}_{k,1,n}, {\psi}_{k,2,n}, \cdots  ,{\psi}_{k,k-1,n},0$ $,{\psi}_{k,k+1,n}, \cdots, {\psi}_{k,{K},n}]^T$, and $\mathbf{e}_k$ is the $k^{th}$ unit vector. Therefore, the LHS of $\textrm{C}_5$ in~\eqref{fg105newa} is written as
\begin{eqnarray} \label{keyh}
\sum_{n=1}^{{N}} \big \{  \log_2 \left({ { \mathbf{q}}_{k,n}^{T}\mathbf{p}_{\mathrm{I},n} + \zeta_{k,n,b} } \right) - \log_2 \left({{ \mathbf{b}}_{k,n}^{T}\mathbf{p}_{\mathrm{I},n} + \zeta_{k,n,b} } \right)  \big \},
\end{eqnarray}
where $ \mathbf{q}_{k,n}= \mathbf{a} _{k,n}+\mathbf{b} _{k,n}$ and $\zeta_{k,n,b} = {\sigma}^2_{n} {\widetilde{{\psi}}}_{k,n}  +{\sigma}^2_{k,n} + {\delta}^2_{k,n}$. Then, similar to the procedure in Subsection~\ref{suma} for $\textrm{C}_5$, we resort to the MM technique. Precisely, considering the inequality in~\eqref{fg103}, the second term in~\eqref{keyh} is minorized by setting $x =  {{ \mathbf{b}}_{k,n}^{T}\mathbf{p}_{\mathrm{I},n} + \zeta_{k,n,b} }$ and $x_{0} =  {{ \mathbf{b}}_{k,n}^{T}\mathbf{p}^{0}_{\mathrm{I},n} + \zeta_{k,n,b} }$. By substituting the minorizer in~\eqref{keyh}, the constraint $\textrm{C}_5$ at the $i^{th}$ iteration is obtained as
\begin{align} \label{kk} 
\textrm{C}_5: \sum_{n=1}^{{N}} \Big \{ & \log_2 \left({{ \mathbf{q}}_{k,n}^{T}\mathbf{p}_{\mathrm{I},n} + \zeta_{k,n,b}} \right) - \log_2 ({{ \mathbf{b}}_{k,n}^{T} \mathbf{p}^{(i-1)}_{\mathrm{I},n} + \zeta_{k,n,b} }  )   \\ \nonumber & - \frac{\log_2 e}{{ \mathbf{b}}_{k,n}^{T}\mathbf{p}^{(i-1)}_{\mathrm{I},n} + \zeta_{k,n,b}} \mathbf{b}^{T}_{k,n} \left( \mathbf{p}_{\mathrm{I},n} -\mathbf{p}^{(i-1)}_{\mathrm{I},n} \right)  \Big \} \geq  \alpha_b.
\end{align}
Next, we consider the non-convex constraint $\textrm{C}_4$. It is observed that the term ${E}_{k} \left(\{ {\mathbf{s}}_{\mathrm{E},n} \}_{n=1}^N \right)$ in the LHS of the this constraint is neither convex nor concave w.r.t. ${\mathbf{s}}_{\mathrm{E},n}$. Therefore, similar to the procedure in Subsection~\ref{suma}, we apply the MM by selecting ${\beta}_{k,n,b}$ (see Appendix~\ref{beta12}) and minorize $\textrm{C}_4$ at the $i^{th}$ iteration by
\begin{align} \label{kkk1}
 & {E}_{k} \left(\left\{ {\mathbf{s}}^{(i-1)}_{\mathrm{E},n} \right\}_{n=1}^N \right) + \frac{1}{2}   \sum_{n=1}^{{N}}  \beta_{k,n,b} {\left( \mathbf{s}^{{(i-1)}}_{\mathrm{E},n}\right)}^{H} \mathbf{s}^{(i-1)}_{\mathrm{E},n}  + \sum_{n=1}^{{N}}    \Re  \left \{ {\hm{\vartheta}}_{k,n,b}^{(i)} \left(\mathbf{s}_{\mathrm{E},n} - {\mathbf{s}}^{(i -1)}_{\mathrm{E},n} \right) \right \} \\ \nonumber & \hspace{5pt}- \frac{1}{2}   \sum_{n=1}^{{N}}  \beta_{k,n,b} {\mathbf{s}^{H}_{\mathrm{E},n}} \mathbf{s}_{\mathrm{E},n} \geq E_{\textrm{min},k} ,  \hspace{5pt} \forall k,
\end{align}
where we define
\begin{align*}
 {\hm{\vartheta}}_{k,n,b}^{(i)}  =& \beta_{k,n,b} \left({\mathbf{s}_{\mathrm{E},n}^{(i-1)}}\right)^{{H}}  + \frac{\tau}{\rho} \textrm{exp}  \widetilde{c}\hspace{0.5mm}  \textrm{exp} \left( \widetilde{a}   {\log}^2  \omega_{k,b}^{(i)}    \right) {\left(    \omega_{k,b}^{(i)}  \right) }^{\widetilde{b}-1}  \left({2\widetilde{a} \log  \omega_{k,b}^{(i)}  } \hspace{2pt}    +\widetilde{b} \right) \hspace{2pt} \left({{\mathbf{s}}^{(i-1)}_{\mathrm{E},n}}\right)^{{H}} \hm{\Xi}_{k,n},
\end{align*}
with $ \omega^{(i)}_{k,b}  = \frac{1}{2} \sum_{n=1}^{{N}} \left({{\mathbf{s}}^{(i-1)}_{\mathrm{E},n}}\right)^{{H}} \hm{\Xi}_{k,n} {\mathbf{s}}^{(i-1)}_{\mathrm{E},n}$. Consequently, the $i^{th}$ iteration of the MM update for~\eqref{fg105newa} is obtained easily as the interior point solution of the following convex problem
\begin{eqnarray} \label{kkk}
	\max_{\alpha_b, \lbrace \mathbf{s}_{\mathrm{E},n}, \mathbf{p}_{\mathrm{I},n} \rbrace ^{{N}}_{n=1}} && \displaystyle \alpha_b \\ \nonumber
	\mbox{s.t.}\;\;&&
	\textrm{C}_2,~\textrm{C}_3,~\textrm{C}_4:\eqref{kkk1},~\textrm{C}_5:\eqref{kk}.
\end{eqnarray}
\subsection{Maximization over $\tau$} \label{minc}
The optimization problem in \eqref{fg10} w.r.t. ${\tau}$ becomes 
\begin{eqnarray}\label{fg190}
	\min_{{\tau}} && {\tau} \\ \nonumber \mbox{s.t.}\;\;&&
	\textrm{C}_1: 0 \leq {\tau} \leq T,~~~
	\textrm{C}_2: {\tau} v_k \leq \widetilde{v}_{k},  \hspace{3pt} \forall k,~~~
	\textrm{C}_3: {\tau} \widehat{v}_1 \leq \widehat{v}_2, ~~~
	\textrm{C}_4: {\tau}   \geq \bar{v}_k ,  \hspace{3pt} \forall k,
\end{eqnarray}
where
\begin{equation}
	v_{k} = \frac{1}{2\rho}  \sum_{n=1}^{{N}} \left \lbrace \vert s_{\mathrm{E},k,n} \vert ^2 -  p_{\mathrm{I},k,n} \right \rbrace  ,~\forall k,~~~\widetilde{v}_{k} = T   \sum_{n=1}^{{N}} \left \lbrace p^{\mathrm{rf}}_{{k},n} -\frac{p_{\mathrm{I},k,n}}{2\rho}  \right \rbrace  ,~\forall k,
\end{equation}
\begin{equation}
	\widehat{v}_{1} = \frac{1}{2\rho}   \sum_{n=1}^{{N}} \left \lbrace \mathbf{s}^{H}_{\mathrm{E},n}  {\mathbf{V}}_{\mathrm{E},n}  \mathbf{s}_{\mathrm{E},n} +\sigma_{n}^{2} \textrm{tr}\left \lbrace \mathbf{U}_{\mathrm{E},n} \mathbf{U}^{H}_{\mathrm{E},n} \right\rbrace   -   \textrm{tr} \lbrace  \mathbf{Q}_{\mathrm{I},n}  { \mathbf{V}}_{\mathrm{I},n} \rbrace -  \sigma_{n}^{2} \textrm{tr} \lbrace \mathbf{U}_{\mathrm{I},n} \mathbf{U}^{H}_{\mathrm{I},n} \rbrace   \right \rbrace  ,
\end{equation}
\begin{equation}
	\widehat{v}_{2} = T  \sum^{{N}}_{n=1} \left \lbrace p^{\mathrm{rf}}_{n} -\frac{1}{2\rho} \left( \textrm{tr}\lbrace  \mathbf{Q}_{\mathrm{I},n}  { \mathbf{V}}_{\mathrm{I},n} \rbrace +  \sigma_{n}^{2} \textrm{tr}\left \lbrace \mathbf{U}_{\mathrm{I},n} \mathbf{U}^{H}_{\mathrm{I},n} \right\rbrace \right)   \right \rbrace  ,
\end{equation}
\begin{equation} \label{jjj}
	\bar{v}_k = \frac{\rho{E}_{\textrm{min},k}}{\textrm{exp} \left( \widetilde{a} \hspace{1pt}  {\textrm{log}}^2   p_{\mathrm{E},k}      \right) p^{  \widetilde{b}}_{\mathrm{E},k} \hspace{1pt} \textrm{exp} \left( \widetilde{c} \right)},~\forall k.
\end{equation}
\begin{algorithm}[t] \label{tt}
	\caption{The Proposed Method for Minimum Rate Maximization in Relay/IRS Systems}
	\begin{algorithmic}[tb]
		\STATE{\!\!\!\!\!\!\!\!\!\!\!\!\!  1.} Relay: Initialize $\mathbf{U}^{(l)}_{\mathrm{E},n}, \mathbf{U}^{(l)}_{\mathrm{I},n} \in \complexC^{{M_{\mathrm{R}}} \times {M_{\mathrm{R}}}}$, $\tau^{(l)} \in \realR_+$, $l\leftarrow0$.
		\STATE{\!\!\!\!\!\!\!\!\!\!\!\!\!  1.} IRS: Initialize ${\hm{\theta}}^{(l)}_{\mathrm{E}}, {\hm{\theta}}^{(l)}_{\mathrm{I}} \in \complexC^{M_{\mathrm{IRS}}}$, $\tau^{(l)} \in \realR_+$,  $l\leftarrow0$.
		\REPEAT
		\STATE{2.} Relay:~Initialize $ \mathbf{U}^{ (i)}_{\mathrm{E},n}$ and $ \mathbf{U}^{ (i)}_{\mathrm{I},n}$ and set $i=0$.
		\STATE{2.} IRS:~Initialize ${\hm{\theta}}^{(i)}_{\mathrm{E}}, {\hm{\theta}}^{(i)}_{\mathrm{I}}$ and set $i=0$.
		\REPEAT
		\STATE{3.} Relay: Solve   (\ref{kkk12}) to obtain $\{ \mathbf{U}_{\mathrm{E},n}, \mathbf{U}_{\mathrm{I},n}, \alpha_a\}$.
		\STATE{3.} IRS: Solve   (\ref{maxmin41}) to obtain $\{ {\hm{\theta}}_{\mathrm{E}}, {\hm{\theta}}_{\mathrm{I}}, \alpha_a\}$.
		\STATE{4.} Update $i \leftarrow i+1$.
		\UNTIL convergence
		\STATE{5.} Relay/IRS:~Initialize $\mathbf{s}^{(i)}_{\mathrm{E},n}$, $\mathbf{p}^{(i)}_{\mathrm{I},n}$ and set $i=0$.
		\REPEAT
		\STATE{6.} Relay/IRS:~Solve the convex problem in (\ref{kkk}) to obtain $ \{\mathbf{s}_{\mathrm{E},n}, \mathbf{p}_{\mathrm{I},n}, \alpha_b\} $.
		\STATE{7.} Update $i \leftarrow i+1$.
		\UNTIL convergence
		\STATE{ 8.}   Relay/IRS:~Compute $\tau^{(l)}$ via the closed-form solution in~\eqref{for1}.
		\STATE{9.} Update $l \leftarrow l+1$.
		\UNTIL convergence
	\end{algorithmic}
\end{algorithm}
Therefore, a closed-form solution (for a non-empty feasible set\footnote{The following conditions lead to a non-empty feasible set for the problem: \\1) ${\bar{v}_k} \leq T , \forall k$, ~2) $\widetilde{v}_{k} \geq 0 , \forall k$, ~3) ${\widehat{v}}_{2} \geq 0$, ~4) $\frac{\widetilde{v}_{j}}{{v}_{j}} {\vert}^{{K}}_{j=1} \geq {\bar{v}_k} , \forall k$ (for ${v}_{j} \geq 0, \forall j$), ~5) $\frac{{\widehat{v}}_{2}}{{\widehat{v}}_{1}} \geq {\bar{v}_k} , \forall k$ (for ${v}_{1} \geq 0$).}) can be obtained as
\begin{equation} \label{for1}
	\tau_{\textrm{opt}}= \textrm{max} \lbrace \bar{v}_1, \bar{v}_2, ... , \bar{v}_{{K}}  \rbrace.
\end{equation}
Algorithm~1 summarizes the discussions in Section~\ref{sum} and represents the steps of the proposed method for maximizing the minimum rate of all user pairs in relay/IRS WPC systems.
Note that similar mathematical derivations are used to develop t-f-static algorithm for relay system as well as t-static algorithm for both relay and IRS systems.
\begin{remark}[convergence]
It has been shown that under some mild conditions, the MM technique converges to the stationary points of the problem~\cite{rezaei2019throughput,naghsh2019max}.
\end{remark}
\subsection{Complexity Analysis} \label{comp}
The main computational burdens in Algorithm 1 are associated with steps~3, 6, and 8. At each inner iteration in step~3, the convex problems in~\eqref{kkk12} and~\eqref{maxmin41} are solved via interior point methods for relay and IRS system design, respectively, with a computational complexity of $\mathcal{O}\left( \left(  2NM^2_{\mathrm{R}} (1+2N)(1+K)  \right)^{3.5} \right)$ and $\mathcal{O} \left( \left( 6M_{\mathrm{IRS}} (N+K+1) \right)^{3.5} \right)$~\cite{ben2001lectures}. Table~\ref{table21} compares the computational complexity of step~3 for other versions of relay/IRS models. Similar to step 3, the complexity (per inner iterations) for step 6  which solves~\eqref{kkk} (e.g., by using the interior point methods) is $\mathcal{O}\left( \left( KN (1+2N)(5+2K) \right)^{3.5} \right)$ for all versions of relay/IRS models. In step 8, the closed-form expression in \eqref{for1} must be calculated leading to the complexity of\footnote{This can be decreased to $\mathcal{O}(N^{2.3})$ via finding the best order of matrix multiplications (see \cite{czumaj1996very} for details).} $\mathcal{O}(N^{3})$.
\begin{table}
 \centering
 \caption{The Computational Complexity Order (per Inner Iterations) for Step 3 of the Algorithm~1.}\label{table21}
 \begin{tabular}{|l|c|}
 \hline
  Relay& $\mathcal{O}\left( \left(  2NM^2_{\mathrm{R}} (1+2N)(1+K)  \right)^{3.5} \right)$ \\ \hline
 Relay (t-static)& $\mathcal{O}\left( \left(  NM^2_{\mathrm{R}} (1+N)(1+2K)  \right)^{3.5} \right)$\\ \hline
 Relay (t-f-static)& $\mathcal{O}\left( \left(  2M^2_{\mathrm{R}} (1+2K)  \right)^{3.5} \right)$\\ \hline
 IRS & $\mathcal{O} \left( \left(6M_{\mathrm{IRS}} (N+K+1) \right)^{3.5} \right)$\\ \hline
 IRS (t-static) & $\mathcal{O} \left( \left(2M_{\mathrm{IRS}} (N+2K+1) \right)^{3.5} \right)$\\
 \hline
 \end{tabular} \label{tab21}
\end{table}
\section{Numerical Examples }\label{num}
Here, we evaluate the
proposed relay/IRS method in different scenarios. The channels from transmitters to the relay and the channels from the relay to the receivers are modeled as $\mathbf{H}_n =0.1 \left( \frac{\widetilde{{d}}_1}{{d}_0} \right)^{\frac{-\widetilde{\gamma}}{2}} \widetilde{\mathbf{H}}_n$ and $\mathbf{G}_n =0.1 \left( \frac{\widetilde{{d}}_2}{{d}_0} \right)^{\frac{-\widetilde{\gamma}}{2}} \widetilde{\mathbf{G}}_n$, respectively, where ${d}_0=1$ m is a reference distance, ${\widetilde{{d}}_1}$ is the distance between $\mathcal{T}_k$ and the relay, ${{\widetilde{d}}_2}$ is the distance between the relay and $\mathcal{R}_k$, and $\widetilde{\gamma}=3$ is the path-loss exponent. It is assumed that the elements of $\widetilde{\mathbf{H}}_n$ and $\widetilde{\mathbf{G}}_n$ are i.i.d. CSCG random variables with zero mean and unit variance. As shown in Fig.~\ref{ht1}, the transmitters and receivers are distributed uniformly within a circle with radius $r_{\mathcal{T}}$ and $r_{\mathcal{R}}$, respectively. We set the distance parameters as ${{d}_1}={{d}_2}={{d}_3}=10$ m and $r_{\mathcal{T}}=r_{\mathcal{R}}=5$ m.
The maximum power budget for $\mathcal{T}_k$, relay, and IRS are set to $p^{\mathrm{rf}}_{k,n}=p^{\mathrm{rf}}_{\mathrm{R},n}=28$ dBm, $p^{\mathrm{rf}}_{\mathrm{IRS}}=20$ dBm, $\forall k,n$,  and the noise power at the relay, IRS and receivers are supposed to be ${\sigma}^2_{\mathrm{R},n}={\sigma}^2_{k,n}= {\delta}^2_{k,n}= -80$ dBm, ${\sigma}^2_{\mathrm{IRS},n}=-100$ dBm, $\forall k,n$. The total bandwidth is fixed to $B_t=1$ MHz. We further assume  the total operation time $T=1$. The curve fitting parameters for non-linear EH circuits are equal to $\widetilde{a}=-0.11$, $\widetilde{b}=-1.17$, and $\widetilde{c}=-12$~\cite{clerckx2018beneficial}. Also, we set ${K}=5$, ${N}=8$, $M_{\mathrm{R}}=6$, and $E_{\textrm{min},k}=E_{\textrm{min}}=10$~$\mu W,~ \forall k$, unless otherwise specified. We solve the convex optimization problems using CVX~\cite{cvx}.
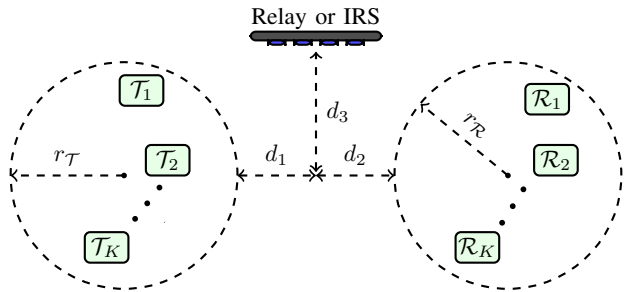
\begin{figure}
 \centering
 \begin{tikzpicture}[even odd rule,rounded corners=2pt,x=12pt,y=12pt,scale=.55,every node/.style={scale=.9}]

 \draw[thick,fill=green!10] (-11.25,4) rectangle ++(2.5,1.75) node[midway]{ $\mathcal{T}_1$};

 \draw[thick,fill=green!10] (-8.5-1.25,0) rectangle ++(2.5,1.75) node[midway]{ $\mathcal{T}_2$};

 \draw[thick,fill=green!10] (-13.25,-5) rectangle ++(2.5,1.75) node[midway]{ $\mathcal{T}_{K}$};


 \draw[thick,fill=black!70] (-3.75,7.75) rectangle ++(7.5,.5);
 \draw[thick,fill=blue!70] (-3.75+1,7.5) rectangle ++(1,.25);
 \draw[thick,fill=blue!70] (-3.75+1+1.5,7.5) rectangle ++(1,.25);
 \draw[thick,fill=blue!70] (-3.75+1+3,7.5) rectangle ++(1,.25);
 \draw[thick,fill=blue!70] (-3.75+1+4.5,7.5) rectangle ++(1,.25);
 \node[] at (0,9) {Relay or IRS};

 \draw[dashed,black!100,thick] (-11,0) circle (6.5) ;
 \draw[fill=black!100,thick] (-11,0) circle (.1) ;
 \draw[dashed,black!100,->,thick] (-11,0) -- (-11-6.5,0) node[pos=0.5,above,sloped]{$r_{\mathcal{T}}$};

 \draw[dashed,black!100,<->,thick] (0,.20) -- (0,7.1) node[pos=0.5,right]{$d_3$};

 \draw[dashed,black!100,<->,thick] (-4.5,0) -- (0,0)
 node[pos=0.5,above,sloped]{$d_1$};
 \draw[dashed,black!100,<->,thick] (0,0) -- (4.5,0)
 node[pos=0.5,above,sloped]{$d_2$};

 \draw[dashed,black!100,thick] (11,0) circle (6.5) ;
 \draw[fill=black!100,thick] (11,0) circle (.1) ;
 \draw[dashed,black!100,->,thick] (11,0) -- (6,4.1) node[pos=0.5,above,sloped]{$r_{\mathcal{R}}$};

 \draw[thick,fill=green!10] (12,3.5) rectangle ++(2.5,1.75) node[midway]{ $\mathcal{R}_1$};

 \draw[black]  (12,-4.5)--(13,0.5) node[pos=0.5,above,sloped] {\huge$\ddots$};
 \draw[white,fill=white]  (11.7,-4.7)--(12.1,-4.7)--(13.1,.6)--(12.9,.6)--(11.7,-4.7);
 \draw [black!100] (-8.7,-2.7)--(-8.25,-1.2) node[pos=0.5,above,sloped] {\huge$\ddots$};
 \draw[white,fill=white]  (-8.8,-2.7)--(-8.6,-2.7)--(-8.15,-1.2)--(-8.35,-1.2)--(-8.8,-2.7);

 \draw[thick,fill=green!10] (12.5,0) rectangle ++(2.5,1.75) node[midway]{ $\mathcal{R}_2$};
 \draw[thick,fill=green!10] (8,-5) rectangle ++(2.5,1.75) node[midway]{ $\mathcal{R}_{K}$};

 \end{tikzpicture}
 \caption{Simulation setup for relay/IRS WPC systems with ${K}$ user pairs.}
 \label{ht1}
 \centering
\end{figure}
\subsection{Relay System}
\begin{table}
 \centering
 \caption{The Baseline Benchmark Methods}\label{table2}
 \begin{tabular}{|l|c|c|}
 	\hline
 &Baseline 1 & Baseline 2\\ \hline
 Information Power Allocation& \ding{51}& \ding{51}\\ \hline
 Energy Waveform Design& \ding{55} & \ding{51} \\ \hline
 Time Allocation& \ding{51} & \ding{51} \\ \hline
 Energy/Information Relay Beamforming & \ding{51} & \ding{55}\\
 \hline
 \end{tabular} \label{tab2}
\end{table}
Here, we compare the results of the proposed algorithms with partially optimized methods (referred to as baseline schemes in the sequel) listed in Table~\ref{tab2}.
For the first baseline method, the energy signals are not optimized, and in the second baseline method, there is no optimization for the relay beamformer; more precisely, the relay amplification matrices are assumed to be identity matrices, i.e. $\mathbf{U}^{\mathrm{R}}_{\mathrm{E},n}=\widetilde{\alpha}_{\mathrm{E},n} \mathbf{I}_{M_{\mathrm{R}}},\forall n,$ and $\mathbf{U}^{\mathrm{R}}_{\mathrm{I},n}=\widetilde{\alpha}_{\mathrm{I},n} \mathbf{I}_{M_{\mathrm{R}}}, \forall n$, where the scalar parameters $\widetilde{\alpha}_{\mathrm{E},n}$ and $\widetilde{\alpha}_{\mathrm{I},n}$ are employed to satisfy the feasible set $\Omega$ in $\eqref{fg10}$.
\begin{figure}
	\hfill
	\subfigure[]{\includegraphics[width=7cm,height=4.25cm]{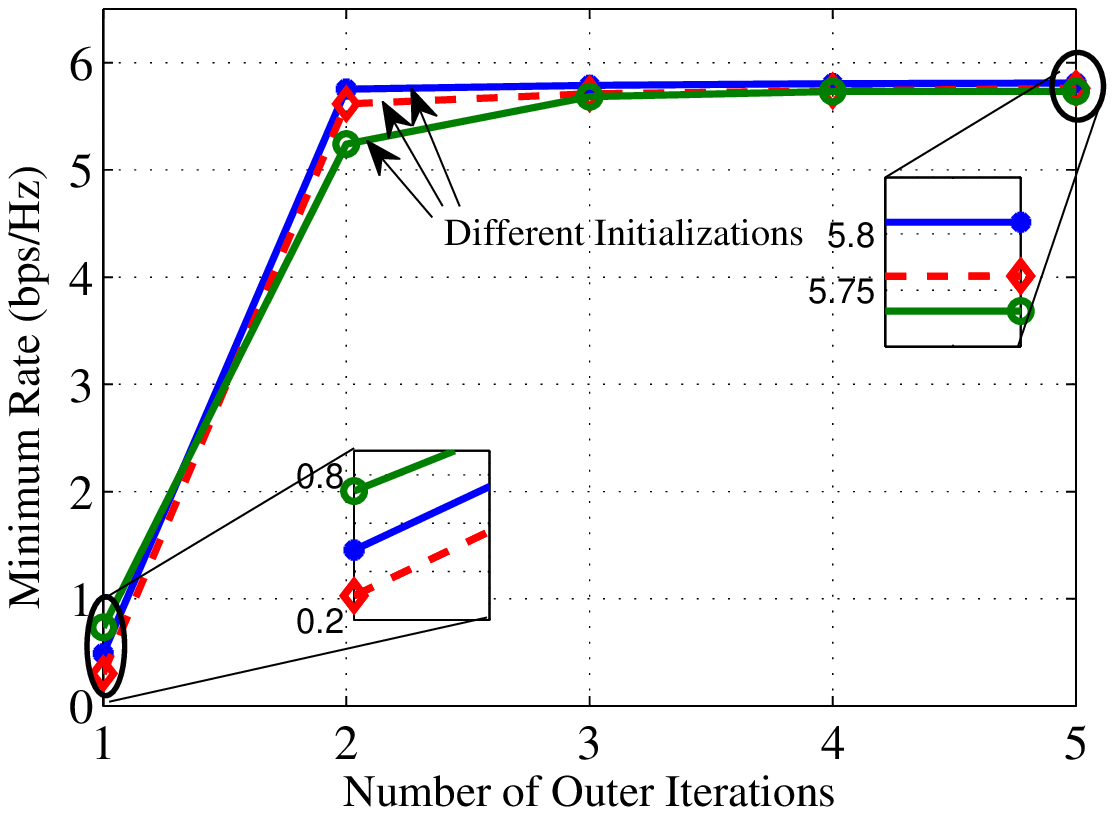}}
	\hfill
	\subfigure[]{\includegraphics[width=4.5cm,height=4.25cm]{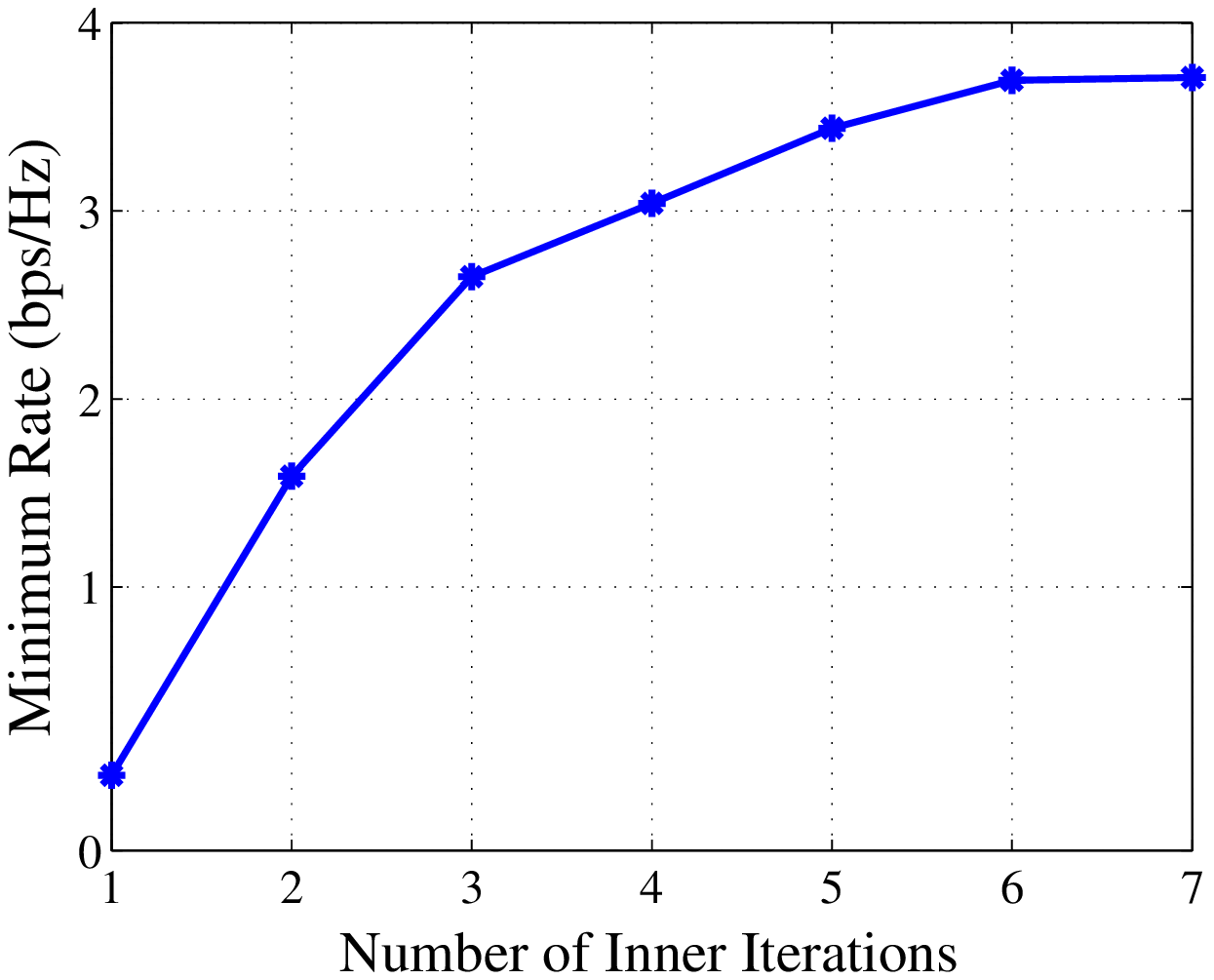}}
	\hfill
	\subfigure[]{\includegraphics[width=4.5cm,height=4.25cm]{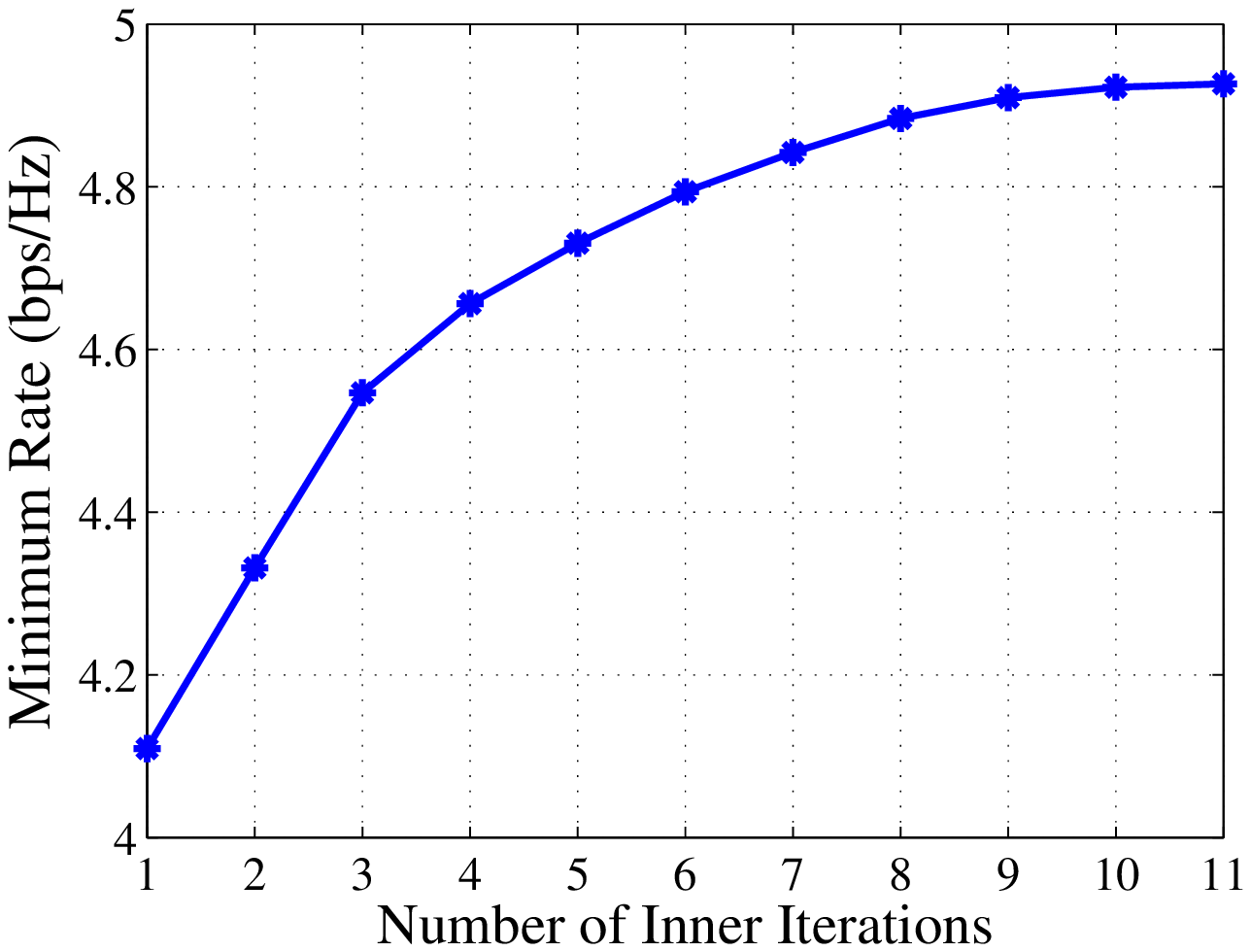}}
	\hfill
	\caption{Convergence behavior of the proposed method in Algorithm~1: (a) outer iterations for three random initial points, (b) inner iterations associated with the sub-problem \ref{relaysub} in the first outer iteration, (c) inner iterations associated with the sub-problem \ref{mina} in the first outer iteration.}
	
	\label{t0}
\end{figure}
The convergence of the proposed algorithm for inner and outer iterations (see Algorithm~1) are plotted in Fig.~\ref{t0}. This figure shows that the proposed algorithm converges within a few outer iterations. Also, in this example, the three different initializations lead to almost the same final value.

In Fig.~\ref{httttttt051}.a, we illustrate the rate-energy region of the proposed method in comparison with the first baseline method for different number of subbands. We can observe that the minimum rate increases as ${N}$ grows.
The optimal time allocation parameter ${\tau}_{\textrm{opt}}$ w.r.t. the EH target is depicted in Fig.~\ref{httttttt051}.b. It is seen that the increased energy threshold $E_{\textrm{min}}$ leads to a larger $\tau$. As $\tau$ increases, the duration of the ID phase decreases. Therefore, as we observe in Fig.~\ref{httttttt051}.a, the minimum rate reduces with increasing $E_{\textrm{min}}$. Also, the impact of the energy waveform design is evident in both figures.
\begin{figure}
\centering
 \subfigure[the rate-energy region]{\includegraphics[width=8cm,height=4.5cm]{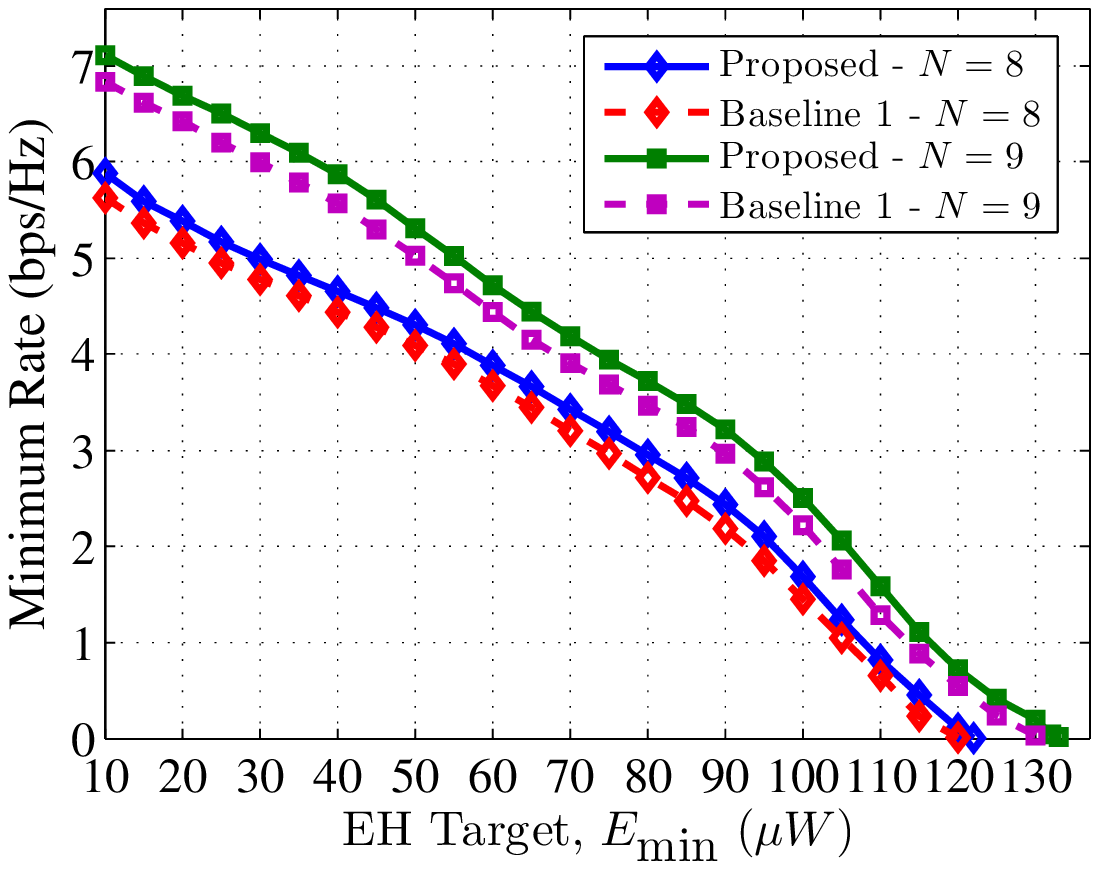}}
 \subfigure[the time allocation parameter $\tau_{\textrm{opt}}$]{\includegraphics[width=8cm,height=4.5cm]{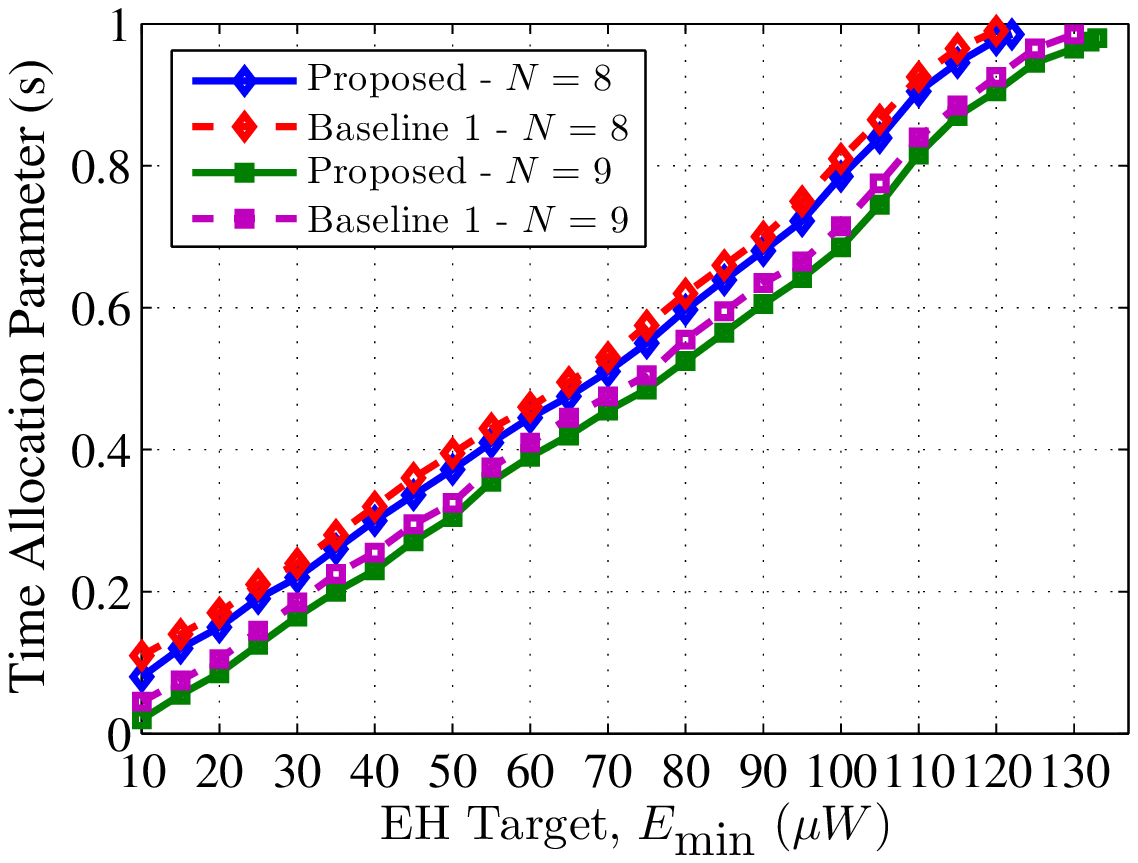}}
 \caption{Comparison of the proposed and baseline 1 methods for different number of subbands $N =8,9$.}
 \label{httttttt051}
\end{figure}
\begin{figure}
 \centering
 \subfigure[minimum rate versus number of pairs $K$]{\includegraphics[width=8cm,height=5cm]{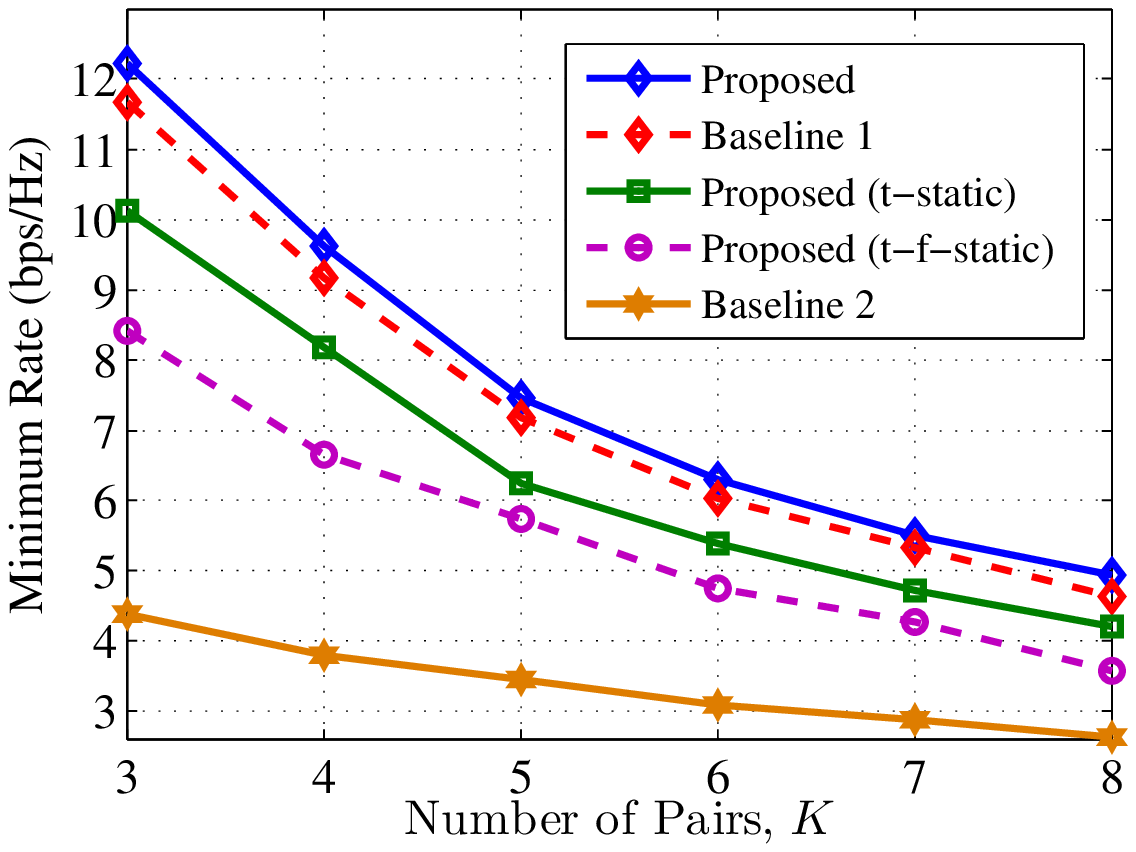}}
 \subfigure[minimum rate versus number of antennas $M_{\mathrm{R}}$]{\includegraphics[width=8cm,height=5cm]{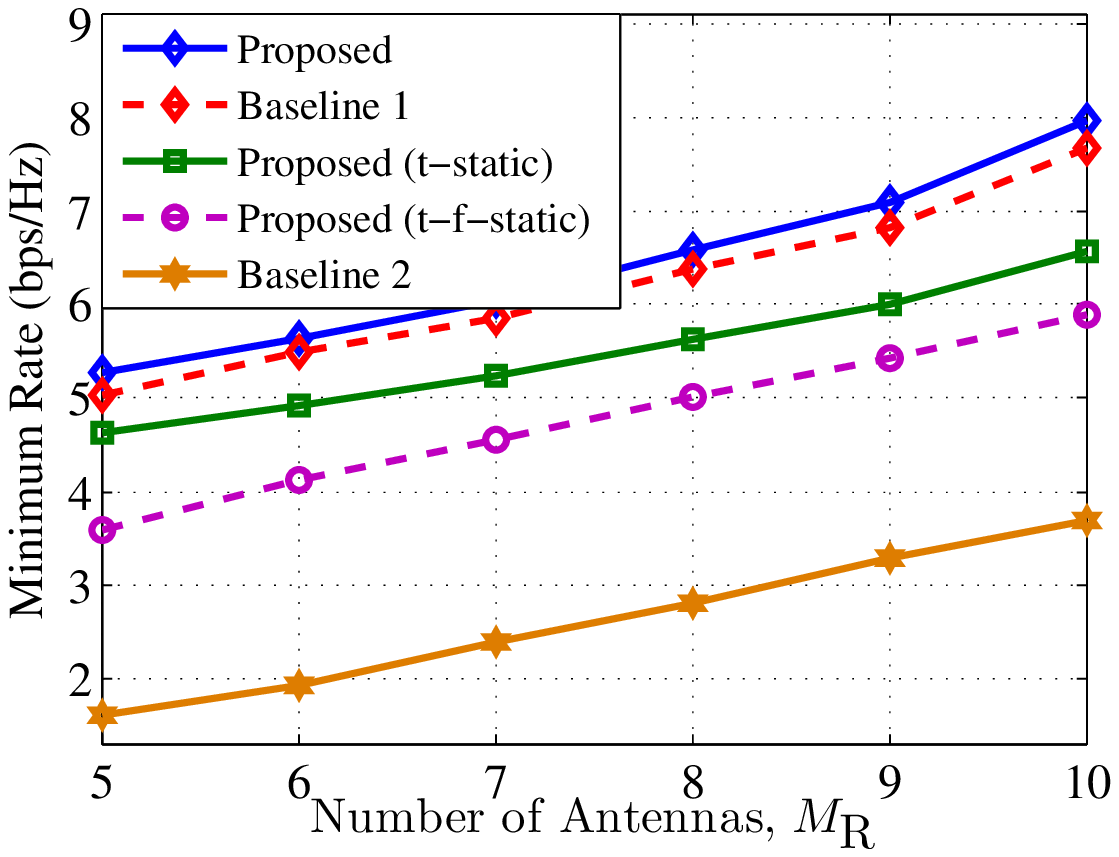}}
 \caption{Comparison of the proposed optimal and sub-optimal methods with baseline methods.}
 \label{httttttt01}
\end{figure}
In Fig.~\ref{httttttt01}.a and Fig.~\ref{httttttt01}.b, we compare the minimum rate of the proposed optimal and sub-optimal approaches with baseline methods. As we can see in Fig.~\ref{httttttt01}.a, increasing the number of pairs results in lower minimum rate for all methods with $M_{\mathrm{R}}=9$.
Furthermore, Fig.~\ref{httttttt01}.b shows that a larger $M_{\mathrm{R}}$ increases the minimum rate with an almost linear trend. The importance of the energy waveform and relay beamforming design is observed through both figures. We can see that the method with no relay beamforming has the worst performance compared to other methods since without a relay amplification matrix design, inter-pair interference cannot be managed.
\subsection{IRS System}
In this subsection, the performance of the proposed IRS-assisted WPC system is evaluated. Since most of the studied scenarios for relay (i.e., Fig.~\ref{t0}, Fig.~\ref{httttttt051}, and Fig.~\ref{httttttt01}.a) have similar trends for IRS, we only consider the scenario of Fig.~\ref{httttttt01}.b, for the sake of brevity. As we can see from Fig.\ref{jhjh}, in the case of IRS, the minimum rate has a super-linear ascent property versus increasing $M_{\mathrm{IRS}}$.
\begin{figure}
 \centering
 \includegraphics[scale=.57]{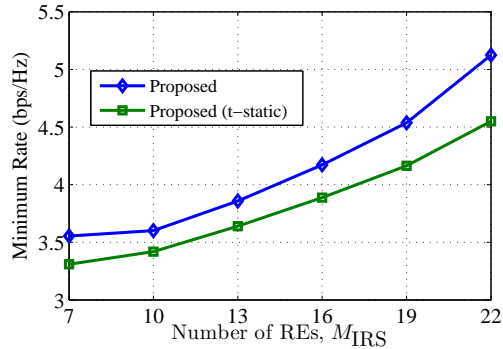}
 \caption{The effect of the number of REs $M_{\mathrm{IRS}}$ for the proposed optimal and sub-optimal IRS methods.}
 \label{jhjh}
 \centering
\end{figure}
\section{Conclusion} \label{con}
In this paper, the max-min rate maximization in a multi-carrier relay/IRS WPC system with a joint TS scheme was considered. A unified framework was proposed to maximize the minimum rate of the user pairs in both relay and IRS systems by jointly designing the energy waveforms, power of information waveforms, amplification matrices, and the time allocation parameter. The non-linearity in EH circuits was also considered in the design problem. The non-convex problem was handled via the MM technique. Numerical results demonstrated the effectiveness of the proposed algorithm in terms of the minimum rate.
As a extended future work in this area, it might be interesting to develop a distributed algorithm for design of multi-user relay/IRS WPC systems.
\appendices
\section{The derivation of the expressions in~\eqref{nnh} and~\eqref{keyA}-\eqref{keyAbar} } \label{app1}
Using
$ \textrm{tr}\left \{ {\mathbf{X}}^{H} {\mathbf{Y}} \right\} = {\textrm{vec} ( \mathbf{X} )}^{H} {\textrm{vec} ( \mathbf{Y} )} $
and
$ \textrm{vec} (\mathbf{X} \mathbf{Y} \mathbf{Z})=  ( \mathbf{Z}^{T} \otimes \mathbf{X} ) \textrm{vec}  ( \mathbf{Y} ) $, the power of the relay signal for ID mode in~\eqref{newq1} can be obtained as
\begin{align*}
  \mathbb{E}\left [{\| \mathbf{\widetilde{r}}(t) \|} _{2}^{2} \right ] 
%
   &  =
   \frac{1}{2} \sum_{n=1}^{{N}}  \bigg(  \mathbf{u}^{H}_{\mathrm{I},n} \textrm{vec} \left( \mathbf{U}_{\mathrm{I},n}  {\mathbf{H}}_n  \mathbf{Q}_{\mathrm{I},n}   {\mathbf{H}}^{H}_n  \right)  +\sigma_{n}^{2}   \mathbf{u}^{H}_{\mathrm{I},n}\mathbf{u}_{\mathrm{I},n}  \bigg)  \\ \nonumber
   &  =
   \frac{1}{2} \sum_{n=1}^{{N}}  \bigg(  \mathbf{u}^{H}_{\mathrm{I},n} \left(  \left( {\mathbf{H}}_n  \mathbf{Q}_{\mathrm{I},n}   {\mathbf{H}}^{H}_n  \right)^{T}  \otimes \mathbf{I}_{M_{\mathrm{R}}} \right)  \mathbf{u}_{\mathrm{I},n}      +\sigma_{n}^{2}   \mathbf{u}^{H}_{\mathrm{I},n}\mathbf{u}_{\mathrm{I},n}  \bigg) 
   \\ \nonumber &   =
\frac{1}{2} \sum_{n=1}^{{N}}    \mathbf{u}^{H}_{\mathrm{I},n} \widetilde{\mathbf{A}}^{\mathrm{R}}_{\mathrm{I},n}   \mathbf{u}_{\mathrm{I},n} .
\end{align*}
Similarly, we can derive the power of the relay signal for the EH mode in~\eqref{bb2} and the expressions in~\eqref{keyA}--\eqref{keyAbar}.
\section{Proof of Lemma 1} \label{app4}
By defining a positive semi-definite matrix $\mathbf{D}$ such that ${\nabla}_{\mathbf{x}}^{2} {s (\mathbf{x})} \preceq \mathbf{D} $, we can write the following majorizer for $s(\mathbf{x})$ as~\cite{sumrate}
\begin{align}\label{app41}
 s(\mathbf{x}) \leq & \hspace{2pt} s(\mathbf{x}_0) + \Re \left \{ \left({{\nabla}_{\mathbf{x}} {s (\mathbf{x})}}\right)^H {\vert}_{\mathbf{x}=\mathbf{x}_0} (\mathbf{x} - \mathbf{x}_0) \right \} + (\mathbf{x} - \mathbf{x}_0)^H  \mathbf{D} (\mathbf{x} - \mathbf{x}_0),
\end{align}
where the gradient and Hessian of ${s (\mathbf{x})}$ are respectively expressed as
\begin{equation} \label{app42}
 {{\nabla}_{\mathbf{x}} {s (\mathbf{x})}}=   \frac{-2\log_2 e}{\mathbf{x}^{H} \mathbf{T} \mathbf{x}  + \nu}  \mathbf{T} \mathbf{x},
\end{equation}
\begin{equation*}
 {\nabla}_{\mathbf{x}}^{2} {s (\mathbf{x})} =  \left( \frac{-2\mathbf{T} }{\mathbf{x}^{H} \mathbf{T} \mathbf{x}  + \nu}    + \frac{4\mathbf{T}  \mathbf{x} \mathbf{x}^{H} \mathbf{T}  }{\left( \mathbf{x}^{H} \mathbf{T} \mathbf{x}  + \nu \right) ^2} \right) \log_2 e.
\end{equation*}
Since $\mathbf{T}  \succeq \mathbf{0}$, the term $\frac{-2\mathbf{T} }{\mathbf{x}^{H} \mathbf{T} \mathbf{x}  + \nu}$ is negative semi-definite, and thus we obtain $\xi > 0$ such that for any $\nu \geqslant 0$
\begin{equation*}
 \frac{4\mathbf{T}  \mathbf{x} \mathbf{x}^{H} \mathbf{T}  }{\left( \mathbf{x}^{H} \mathbf{T} \mathbf{x}  + \nu \right) ^2} \log_2 e \leqslant \frac{4\mathbf{T}  \mathbf{x} \mathbf{x}^{H} \mathbf{T}  }{\left( \mathbf{x}^{H} \mathbf{T} \mathbf{x}  \right) ^2} \leqslant \xi \mathbf{I}_{M^2_{\mathrm{R}}}.
\end{equation*}
Also, as $\mathbf{T}  \mathbf{x} \mathbf{x}^{H} \mathbf{T}$ is a rank-one matrix, we can choose $\xi$ as $\xi \geqslant 4 \phi $, where $\phi$ is given as
\begin{equation} \label{pi}
 \phi= \displaystyle \max_{\mathbf{x}}  \hspace{5pt} \frac{   \mathbf{x}^{H} {\mathbf{T}}^2 \mathbf{x} }{\left( \mathbf{x}^{H} \mathbf{T} \mathbf{x}   \right) ^2}.
\end{equation}
Then by choosing $\mathbf{a}=\mathbf{V}^H \mathbf{x}$, where $\mathbf{V}$ is a full-rank matrix such that $\mathbf{T}=\mathbf{V} \mathbf{V}^H$, the following optimization is equivalently obtained from~\eqref{pi} as
\begin{equation}
 \phi= \displaystyle \max_{\mathbf{a}}  \hspace{5pt} \frac{\mathbf{a}^{H} \mathbf{V}^{H}  \mathbf{V}   \mathbf{a}  }{\mathbf{a}^{H} \mathbf{a}}  \frac{1}{\mathbf{a}^{H} \mathbf{a}} .
\end{equation}
Using $\mathbf{x}^H \mathbf{Q} \mathbf{x} \leq P$ and applying a similar procedure in
\cite[Appendix B]{sumrate},
we can write
\begin{equation*}
 \phi \leq  \frac{ P \lambda _{\textrm{max}} (\mathbf{T})}{\mathbf{v}^{H}_{1} \mathbf{V}^{-1}  \mathbf{Q} \mathbf{V}^{-H}  \mathbf{v}_{1}  }  ,
\end{equation*}
where $\mathbf{v}_{1}$ is the principal eigenvector of $\mathbf{V}^{H}  \mathbf{V}$. Finally, from~\eqref{app41},~\eqref{app42}, and $\xi= 4 \phi$, we obtain
$
 \mathbf{b}={{\nabla} {s (\mathbf{x})}} {\vert}_{\mathbf{x}=\mathbf{x}_0}=   \frac{-2\log_2 e}{\mathbf{x}^{H}_{0} \mathbf{T} \mathbf{x}_{0} + \nu}  \mathbf{T} \mathbf{x}_{0}$ and $
 \mathbf{D}= \frac{4P}{\mathbf{w}^{H}_{1} \mathbf{Q} \mathbf{w}_1 }   \mathbf{I}_{M^2_{\mathrm{R}}},
$
where $\mathbf{w}_1$ is the principal eigenvector of $\mathbf{T}$.
\section{A selection of $\beta_{k,n,a}$ and  $\beta_{k,n,b}$} \label{beta12}
The value of $\beta_{k,n,b}$ should be selected such that ${\nabla}^{2}_{{\mathbf{s}}_{\mathrm{E},n}} {E}_{k} \left( \{{\mathbf{s}}_{\mathrm{E},n}\}_{n=1}^N \right) + {\beta}_{k,n,b} \mathbf{I}_{{{K}}} \succeq \mathbf{0}$. The term ${\nabla}^{2}_{{\mathbf{s}}_{\mathrm{E},n}} {E}_{k} \left( \{{\mathbf{s}}_{\mathrm{E},n}\}_{n=1}^N \right)$ is straightforwardly calculated as
\begin{align}
{\nabla}^{2}_{{\mathbf{s}}_{\mathrm{E},n}} {E}_{k} \left( \{{\mathbf{s}}_{\mathrm{E},n}\}_{n=1}^N \right) =&\varrho_k \sum_{n=1}^{{N}} \hm{\Xi}_{k,n} + \eta_k \sum_{n=1}^{{N}} \sum_{n^{\prime}=1}^{{N}}  \hm{\Xi}_{k,n} \mathbf{s}_{\mathrm{E},n} \mathbf{s}^{H}_{\mathrm{E},n^{\prime}} \hm{\Xi}_{k,n^{\prime}},
\end{align}
where
\begin{equation*}
\varrho_k = \frac{\tau}{\rho}  \textrm{exp}{\widetilde{c}} \hspace{2pt} \textrm{exp} \left( \widetilde{a}  {\log}^2   p_{\mathrm{E},k}       \right) p^{  \widetilde{b}-1}_{\mathrm{E},k} \left(  2 \widetilde{a}  \log    p_{\mathrm{E},k}  \hspace{2pt} + \widetilde{b} \right),
\end{equation*}
\begin{align*}
\eta_k =& \frac{\tau \textrm{exp}{\widetilde{c}}\hspace{2pt} \textrm{exp} \left({ \widetilde{a}   {\log}^2    p_{\mathrm{E},k} }\right) p^{  \widetilde{b}-2}_{\mathrm{E},k} }{\rho} \Big( 4 {\widetilde{a}}^2  {\log}^2    p_{\mathrm{E},k}  + \left( 4{\widetilde{a}}{\widetilde{b}} -2{\widetilde{a}} \right) \log  \hspace{0.5mm}  p_{\mathrm{E},k}   \hspace{2pt}  + {\widetilde{b}}^2 -{\widetilde{b}} + 2 \widetilde{a} \Big).
\end{align*}
As $\widetilde{a}<0, \widetilde{b} <0,\hm{\Xi}_{k,n}\succeq \mathbf{0},$ and $\hm{\Xi}_{k,n} \mathbf{s}_{\mathrm{E},n} \mathbf{s}^{H}_{\mathrm{E},n} \hm{\Xi}_{k,n}\succeq \mathbf{0}$, it suffices to choose $\beta_{k,n,b}$ such that
\begin{align} \label{f0}
{\beta}_{k,n,b} \mathbf{I}_{{{K}}} \succeq& - \widetilde{\varrho}_k \sum_{n=1}^{{N}} \hm{\Xi}_{k,n} - \widetilde{\eta}_k \sum_{n=1}^{{N}} \sum_{n^{\prime}=1}^{{N}}  \hm{\Xi}_{k,n} \mathbf{s}_{\mathrm{E},n} \mathbf{s}^{H}_{\mathrm{E},n^{\prime}} \hm{\Xi}_{k,n^{\prime}},
\end{align}
where
\begin{align*}
\widetilde{\varrho}_k =  \frac{\tau}{\rho}  \widetilde{b} \hspace{2pt} \textrm{exp}{\widetilde{c}}\hspace{2pt} \textrm{exp} \left({ \widetilde{a}   {\log}^2    p_{\mathrm{E},k} }\right) p^{  \widetilde{b}-1}_{\mathrm{E},k} ,
\end{align*}
\begin{align*}
\widetilde{\eta}_k =& \frac{\tau}{\rho} \textrm{exp}{\widetilde{c}}\hspace{2pt} \textrm{exp} \left({ \widetilde{a}   {\log}^2    p_{\mathrm{E},k} }\right) p^{  \widetilde{b}-2}_{\mathrm{E},k}  \left( \log    p_{\mathrm{E},k}    \left( 4{\widetilde{a}}{\widetilde{b}} -2{\widetilde{a}} \right) + 2 \widetilde{a} \right).
\end{align*}
Thus from~\eqref{41}, we can write
\begin{equation} \label{fi}
\|{\mathbf{s}}_{\mathrm{E},n}\|^2_2 \leq \frac{2\rho T}{{\tau}} \sum_{k=1}^{{K}} p^{\mathrm{rf}}_{k,n}.
\end{equation}
Finally, using~\eqref{fg10},~\eqref{f0},~\eqref{fi} and knowing  that $\mathbf{s}^{H}_{\mathrm{E},n} \hm{\Xi}_{k,n} \mathbf{s}_{\mathrm{E},n} \leq \|{\mathbf{s}}_{\mathrm{E},n}\|^2_2 \lambda_{\textrm{max}} \left(\hm{\Xi}_{k,n} \right)$, we can select ${\beta}_{k,n,b} > {\beta}^{t}_{k,n,b} $ where 
\begin{align*}
{\beta}^{t}_{k,n,b}=&-  \frac{\tau}{\rho}   \textrm{exp}{ \widetilde{c}}  \hspace{2pt} \textrm{exp}\left(   2\widetilde{a}   {\log}^2   T  \sum_{n=1}^{{N}}  \lambda_{\textrm{max}} \left(\hm{\Xi}_{k,n} \right)  \sum_{k=1}^{{K}} p^{\mathrm{rf}}_{k,n}      \right) \widetilde{f}^{  \widetilde{b}-2}_{k}
\Bigg(
 \left( \left( 4{\widetilde{a}}{\widetilde{b}} -2{\widetilde{a}} \right) \log   \widetilde{f}_{k} \hspace{2pt}  + 2 \widetilde{a} \right) \\ \nonumber & \times \sum_{n=1}^{{N}} \sum_{n^{\prime}=1}^{{N}}  \lambda_{\textrm{max}} \left(\hm{\Xi}_{k,n} \hm{\Xi}_{k,n^{\prime}} \right) \sum_{k=1}^{{K}} \sqrt{p^{\mathrm{rf}}_{k,n} p^{\mathrm{rf}}_{k,n^{\prime}}}
+\widetilde{b} \hspace{2pt} \widetilde{f}_{k} \sum_{n=1}^{{N}} \lambda_{\textrm{max}} \left(\hm{\Xi}_{k,n} \right)
\Bigg),
\end{align*}
with $
\widetilde{f}_{k} = \exp\bigg( {  \frac{-\widetilde{b}-\sqrt{\widetilde{b}^2 -4  \widetilde{a}\hspace{2pt} \log  \frac{\rho E_{\textrm{min},k}}{{\tau} \textrm{exp}{\widetilde{c}}}   } }{2\widetilde{a}}  }\bigg).
$
We can take similar steps for selecting ${\beta}^{t}_{k,n,a}$.
\section{Proof of Lemma~2} \label{app1000}
The ID part of the relay power constraint in~\eqref{bb2} is
$
\mathbf{u}^{H}_{\mathrm{I},n} \widetilde{\mathbf{A}}^{\mathrm{R}}_{\mathrm{I},n}   \mathbf{u}_{\mathrm{I},n}.
$
Only  $(iM_{\mathrm{IRS}}+i+1)^{th},~ 0 \leq i \leq M_{\mathrm{IRS}}-1$ entries of $\mathbf{u}_{\mathrm{I},n}=\textrm{vec}(\textrm{Diag}({\hm{\theta}}_{\mathrm{I}}))$ are non-zero for the IRS system. Thus, we can rewrite $
\mathbf{u}^{H}_{\mathrm{I},n} \widetilde{\mathbf{A}}^{\mathrm{R}}_{\mathrm{I},n}   \mathbf{u}_{\mathrm{I},n}
$ for IRS system  as ${\hm{\theta}}^H_{\mathrm{I}} \widetilde{\mathbf{A}}^{\mathrm{IRS}}_{\mathrm{I},n}   {\hm{\theta}}_{\mathrm{I}}$, where $\widetilde{\mathbf{A}}^{\mathrm{IRS}}_{\mathrm{I},n}$ contains only the $(kM_{\mathrm{IRS}}+k+1,lM_{\mathrm{IRS}}+l+1)^{th},~ 0\leq k,l \leq M_{\mathrm{IRS}}-1$ entries of $\widetilde{\mathbf{A}}_{\mathrm{I},n}$ which is the same as $\widetilde{\mathbf{A}}^{\mathrm{R}}_{\mathrm{I},n}$ with replacing $M_{\mathrm{R}}$ by $M_{\mathrm{IRS}}$. Therefore, from~\eqref{nnh} and by using some matrix manipulations, we obtain
\begin{equation}
 \widetilde{\mathbf{A}}^{\mathrm{IRS}}_{\mathrm{I},n} =   {\left( {\mathbf{H}}_n \mathbf{Q}_{\mathrm{I},n} {\mathbf{H}}^{H}_n \right)}^{T} \odot \mathbf{I}_{{M_{\mathrm{IRS}}}}  +\sigma_{n}^{2} \mathbf{I}_{{M_{\mathrm{IRS}}}}.
\end{equation}
Other expressions in~\eqref{nnh1}-\eqref{keyAbar1} are similarly obtained.
\bibliographystyle{IEEETran}
\bibliography{myreff}
\end{document}